\documentclass[letterpaper,11pt]{article}
\pdfoutput=1
\usepackage[letterpaper,margin=1in]{geometry}
\usepackage[utf8]{inputenc}
\usepackage[english]{babel}
\usepackage[numbers,sort&compress]{natbib}
\usepackage{microtype}
\usepackage[small]{caption}
\usepackage{xcolor}
\usepackage{graphicx}
\usepackage{amsmath}
\usepackage{amsthm}
\usepackage{amssymb}
\usepackage{algorithm2e}
\usepackage{enumitem}
\usepackage{booktabs}
\usepackage{bm}
\usepackage{calc}
\usepackage[unicode]{hyperref}
\usepackage{doi}

\hypersetup{
	colorlinks=true,
	linkcolor=black,
	citecolor=black,
	filecolor=black,
	urlcolor=[rgb]{0,0.1,0.5},
	pdftitle={Truly Tight-in-$\Delta$ Bounds for Bipartite Maximal Matching and Variants},
	pdfauthor={Sebastian Brandt, Dennis Olivetti}
}

\newtheorem{theorem}{Theorem}[section]
\newtheorem*{theorem*}{Theorem}
\newtheorem{lemma}[theorem]{Lemma}

\newtheorem{corollary}[theorem]{Corollary}
\newtheorem{observation}[theorem]{Observation}

\newtheorem{op}[theorem]{Open Problem}
\theoremstyle{definition}
\newtheorem{definition}[theorem]{Definition}

\newtheorem{convention}[theorem]{Convention}
\theoremstyle{remark}

\newtheoremstyle{repeatdefinition}{\topsep}{\topsep}{}{}{\bfseries}{.}{ }{\thmname{#1}\thmnote{ \bfseries #3}}
\theoremstyle{repeatdefinition}

\DeclareMathOperator{\poly}{poly}

\DeclareMathOperator{\REW}{\mathcal R_W}
\DeclareMathOperator{\REB}{\mathcal R_B}

\newcommand{\s}{\mspace{1mu}}
\newcommand{\mybox}[1]{\mspace{2mu}{\setlength{\fboxsep}{1.5pt}\color{lightgray}\boxed{\color{black}\scriptstyle #1}}\mspace{2mu}}
\newcommand{\M}{\mathsf{M}}
\renewcommand{\P}{\mathsf{P}}
\renewcommand{\O}{\mathsf{O}}
\newcommand{\A}{\mathsf{A}}
\newcommand{\B}{\mathsf{B}}
\newcommand{\X}{\mathsf{X}}
\newcommand{\Y}{\mathsf{Y}}
\newcommand{\Z}{\mathsf{Z}}

\newcommand{\bM}{\mybox{\M}}

\newcommand{\bX}{\mybox{\X}}
\newcommand{\bMX}{\mybox{\M\X}}

\newcommand{\bOX}{\mybox{\O\X}}

\newcommand{\bPOX}{\mybox{\P\O\X}}

\newcommand{\bMPOX}{\mybox{\M\P\O\X}}
\newcommand{\bMZPOX}{\mybox{\M\Z\P\O\X}}
\newcommand{\bXYZ}{\mybox{\X\Y\Z}}
\newcommand{\bA}{\mybox{\A}}
\newcommand{\bB}{\mybox{\B}}
\newcommand{\bAB}{\mybox{\A\B}}

\newcommand{\mm}{MM}
\newcommand{\bmm}{BMM}

\newenvironment{myabstract}
{\list{}{\listparindent 1.5em%
		\itemindent    \listparindent
		\leftmargin    1cm
		\rightmargin   1cm
		\parsep        0pt}%
	\item\relax}
{\endlist}

\newenvironment{mycover}
{\list{}{\listparindent 0pt
		\itemindent    \listparindent
		\leftmargin    1cm
		\rightmargin   1cm
		\parsep        0pt}%
	\raggedright
	\item\relax}
{\endlist}

\newcommand{\myemail}[1]{\,$\cdot$\, {\small #1}}
\newcommand{\myaff}[1]{\,$\cdot$\, {\small #1}\par\smallskip}

\begin{document}

	\begin{mycover}
		{\huge\bfseries\boldmath Truly Tight-in-$\Delta$ Bounds for Bipartite Maximal Matching and Variants \par}
		\bigskip
		\bigskip
		\bigskip
		
		\textbf{Sebastian Brandt}
		\myemail{brandts@ethz.ch}
		\myaff{ETH Zurich}
		
		\textbf{Dennis Olivetti\footnote{Part of this work was done while this author was supported by Aalto University, and the Academy of Finland, Grant 285721.}}
		\myemail{dennis.olivetti@cs.uni-freiburg.de}
		\myaff{University of Freiburg}
		
	\end{mycover}
	\bigskip

\begin{myabstract}
	In a recent breakthrough result, Balliu et al.\ [FOCS'19] proved a deterministic $\Omega(\min(\Delta,\log n /\,\log \log n))$-round and a randomized $\Omega(\min(\Delta,\log \log n/\,\log \log \log n))$-round lower bound for the complexity of the bipartite maximal matching problem on $n$-node graphs in the LOCAL model of distributed computing.
	Both lower bounds are asymptotically tight as a function of the maximum degree $\Delta$.

	We provide \emph{truly tight} bounds in $\Delta$ for the complexity of bipartite maximal matching and many natural variants, up to and including the additive constant.
	As a by-product, our results yield a considerably simplified version of the proof by Balliu et al.
	
	We show that our results can be obtained via \emph{bounded automatic round elimination}, a version of the recent automatic round elimination technique by Brandt [PODC'19] that is particularly suited for automatization from a \emph{practical} perspective.
	In this context, our work can be seen as another step towards the automatization of lower bounds in the LOCAL model.
\end{myabstract}

\thispagestyle{empty}
\setcounter{page}{0}
\newpage

\section{Introduction}
The \emph{maximal matching (MM)} problem has been studied extensively in the field of distributed graph algorithms.
In the LOCAL model of distributed computing \cite{Linial1992,Peleg2000}, the recent breakthrough by \citet{Balliu2019} provided lower bounds for the complexity of \mm{} that are asymptotically tight in the maximum degree $\Delta$.
More precisely, for $n$-node graphs, the authors prove lower bounds of $\Omega(\min(\Delta,\frac{\log n}{\log \log n}))$ rounds for deterministic algorithms and $\Omega(\min(\Delta,\frac{\log \log n}{\log \log \log n}))$ rounds for randomized algorithms, while an upper bound of $O(\Delta + \log^* n)$ is known since almost two decades due to a result by \citet{panconesi01simple}.

In other words, it is possible to solve \mm{} in linear-in-$\Delta$ time by paying a small dependency on $n$, and we cannot do better as a function of $\Delta$, unless we pay a high dependency on $n$. The lower bound results have been obtained by using the so-called \emph{round elimination} technique to show that on infinite regular $2$-colored trees the \mm{} problem requires $\Omega(\Delta)$ rounds. This way of proving the lower bounds guarantees that they hold even for the \emph{bipartite maximal matching (BMM)} problem.

\paragraph{Our contributions}
In this work, we prove \emph{truly tight} bounds for the complexity of \bmm{}, i.e., we prove that \emph{exactly} $2 \Delta - 1$ rounds are required.
Moreover, we also prove tight bounds for natural variants of \mm{} by showing that our lower bound technique is robust to changes in the description of the considered problem, and providing optimal algorithms to obtain tight upper bounds.
As a by-product, we obtain a much simplified version of the proof for the celebrated lower bounds presented in \  \cite{Balliu2019}, both in terms of the technical issues and the intuition behind the proof.
Finally, we consider our work as an important step towards the \emph{automatization} of lower bounds: we introduce the notion of \emph{bounded automatic round elimination}, a special variant of the round elimination technique amenable to automatization, and show that our bounds can be obtained in a (semi-)automatable fashion via bounded automatic round elimination.
This also provides another step in better understanding the round elimination technique---a tool that is responsible for most of the lower bounds in the LOCAL model \cite{Brandt2019,Balliu2019hardness,Linial1992, Naor1991,Balliu2019,Brandt2016,chang16exponential,binary,chang18complexity}, but still poorly understood.

\paragraph{Maximal matching and its variants}
In this work, we will consider the following natural problem family (in the bipartite setting).
\begin{definition}($x$-maximal $y$-matching)
	Given a graph $G=(V,E)$, a set $M \subseteq E$ is an $x$-maximal $y$-matching if the following conditions hold:
	\begin{itemize}
		\item Every node is incident to at most $y$ edges of $M$;
		\item If a node $v$ is not incident to any edge of $M$, then at least $\min\{ \deg(v),\Delta - x \}$ neighbors of $v$ are incident to at least one edge of $M$.
	\end{itemize}
\end{definition}
This family of problems contains \mm{} (by setting $x=0$ and $y=1$), and many interesting variants of it, obtained by relaxing the covering or the packing constraint of \mm{}. While the results presented in \cite{Balliu2019} showed asymptotically tight bounds for \mm{} as a function of $\Delta$, no tight bounds are provided for relaxed variants of matchings.

How much easier does \mm{} become if we allow each node to be matched with a constant number of neighbors, instead of just one? We know, for example, that in the non-bipartite setting $\Omega(\log^* n)$ rounds are required to solve \mm{} \cite{Linial1992}, while as soon as we allow nodes to have two incident edges in the matching, the problem becomes solvable without any dependency on $n$ \cite{suomelabook}, but it is not clear how this affects the dependency on $\Delta$. How much easier does the problem become if we don't require strict maximality? We will address these kinds of questions in this work and prove \emph{truly tight} bounds for \mm{} and the whole family of $x$-maximal $y$-matching problems, in the bipartite setting. Note that tightness results of this kind require that the considered problems can be solved independently of $n$, which is the reason for our restriction to the bipartite setting.

\paragraph{Round elimination}
In order to prove our lower bounds, we will make use of the round elimination technique, which works as follows. 
Start from a locally checkable problem $\Pi$, i.e., a problem where all nodes must output labels from some finite set, subject to some local constraints. Try to define a problem $\Pi'$ that is at least one round easier than $\Pi$ (i.e., can be solved strictly faster in the LOCAL model).  Then, repeat the process.
If we prove that the result that we get after $T$ steps of this process cannot be solved in $0$ rounds of communication, then we directly obtain a lower bound of $T+1$ rounds for the original problem $\Pi$.

This is \emph{not} a new technique: it has been used by Linial roughly 30 years ago to prove a lower bound for $3$-coloring a cycle \cite{Linial1992}.
However, after Linial's result, the technique was apparently forgotten until it re-emerged in 2016, when it was used to prove lower bounds for sinkless orientation, $\Delta$-coloring, and the Lov\'asz Local Lemma \cite{Brandt2016,chang16exponential,chang18complexity}.
Since then, round elimination has been used to prove lower bounds for a number of different problems \cite{Brandt2019,Balliu2019hardness,Balliu2019,binary}.
More importantly, though, this technique can be automated in a certain sense.

\paragraph{Automatic round elimination}
In a recent breakthrough result, \citet{Brandt2019} showed that, under certain conditions, given any locally checkable problem $\Pi$, we can \emph{automatically} define a problem $\Pi'$ that is exactly one round easier than $\Pi$.
Conceptually, this tremendously simplifies the task of applying round elimination to prove lower bounds in the LOCAL model: instead of having to find each subsequent problem that is one round easier than the previous one by hand, we can do the same by just mechanically applying certain operations.
The main issue with automatic round elimination, and the reason why the result by Brandt does not immediately provide new lower bounds for all kinds of problems, is the growing \emph{description complexity} inherent in each of the round elimination steps:
If we start from a problem $\Pi$ defined via a set of labels $L$ (and certain output constraints), then the problem $\Pi'$ obtained after one step is defined via a subset of the label set $2^L$.
In other words, even if $\Pi$ is some simple problem that can be described in a compact way, the description of the problem $\Pi'$ can be exponentially larger, and it may be difficult for a human being to understand and feasibly deal with the problems obtained after applying only a few of these steps.

The main approach for dealing with this issue is to find good \emph{relaxations}, i.e., to reduce the description complexity of the obtained problem $\Pi'$ without losing too much round complexity: the goal is to transform $\Pi'$ into a problem $\Pi^*$ such that on one hand, $\Pi^*$ has a much simpler description than $\Pi'$, while on the other hand, $\Pi^*$ is provably at least as easy as $\Pi'$, but not much easier.
If we can find such a relaxation after each round elimination step (and use the relaxed problem as the starting point for the next step), then the number of steps until we obtain a $0$-round solvable problem still constitutes a lower bound for the problem.

\paragraph{Bounded automatic round elimination}
One natural way to ensure that problems do not grow beyond some description complexity threshold is to fix a constant $c$, and after each step of automatic round elimination, relax the problem in a way that ensures that the obtained problem uses at most $c$ labels in its description.

We propose the study of round elimination lower bounds that only require a constant number of labels as a major research program.
Understanding better for which problems we can obtain lower bounds using this technique, and how good the lower bounds achieved in such a way can be, would not only help us to get a better handle on general automatic round elimination, but also has another advantage: the bound on the number of labels ensures (in some sense) that the lower bound can be found \emph{automatically} by a computer.
The basic idea is that while the automatic round elimination technique can be used in theory to define a problem that is exactly $k$ rounds easier than the original problem, for any $k$, this may be really hard to do in practice, since in each step the problem that we obtain can have an exponentially larger description than the old one.
However, if we bound the number of labels to some (reasonably small) constant $c$, a computer can actually try all possible relaxations that give problems with at most $c$ labels.
If at least one obtained relaxation results in a sufficiently hard problem, by repeating this process we can obtain lower bounds automatically also in practice.
We discuss this automatization in more detail in Section \ref{sec:backstage}.

The bottom line is that understanding for which problems a round elimination proof with a restricted number of labels works would allow us to decide which problems to attack with the help of computers and for which problems this would simply be a waste of resources. 
Moreover, a more fine-grained understanding of which constant $c$ (depending on the chosen problem) is required to obtain the largest possible lower bound, or more generally, how the obtained lower bound depends on the chosen constant $c$, would help to direct the use of resources in the most efficient way.

A very interesting example case is provided by the recent $\Omega(\Delta)$-round\footnote{Note that bounds obtained by round elimination are usually bounds in $\Delta$. In order to transform these bounds into bounds in $n$, we simply have to consider graphs with suitably chosen $\Delta$ (as a function of $n$).} lower bound for \bmm{} by \citet{Balliu2019}.
Here, the authors use automatic round elimination with label number restricted by a constant as part of their proof, but this part only yields a lower bound of $\Omega(\sqrt{\Delta})$ rounds, that is subsequently lifted to a bound of $\Omega(\Delta)$ rounds by applying another technique on top of it.
One important ingredient in our new lower bounds is to increase the chosen constant from $4$ to $5$, which results in a direct $\Omega(\Delta)$-round lower bound.
Besides proving that bounded automatic round elimination can yield the full $\Omega(\Delta)$-round lower bound, our result highlights how sensitive the achieved lower bound can be to the exact number of used labels.
While the change from $4$ to $5$ labels increases the difficulty in finding good relaxations after each round elimination step due to the substantially increased number of possible relaxations, the behavior of the obtained problem sequence is actually much easier to understand than the one given in \cite{Balliu2019}.
Together with the fact that we can omit the step in \cite{Balliu2019} that lifts their $\Omega(\sqrt{\Delta})$-round lower bound to $\Omega(\Delta)$ rounds, we obtain a considerably simplified proof for both the deterministic $\Omega(\min(\Delta,\frac{\log n}{\log \log n}))$-round and the randomized $\Omega(\min(\Delta,\frac{\log \log n}{\log \log \log n}))$-round lower bound.

\paragraph{Upper bounds}
Automatic round elimination is not only a tool to prove lower bounds; it can also be used to prove \emph{upper} bounds.
While, for technical reasons, the obtained upper bounds only hold on high-girth graphs, they can still give valuable insight into how a problem could possibly be solved in general.
Furthermore, studying \emph{bounded} automatic round elimination is also very interesting from an upper bound perspective, due to the fact that the bounded number of labels ensures that the resulting algorithm is \emph{bandwidth efficient}, i.e., we do not only obtain an upper bound automatically, but the bound directly applies to the CONGEST model as well!
While we will not prove it, we remark that the upper bounds that we provide in this work can be obtained using bounded automatic round elimination.

\paragraph{Our results}
We prove tight bounds for $x$-maximal $y$-matchings. Let \[
T_\Delta(x,y) =
\begin{cases}
2 \lceil (\Delta - x)/y\rceil, ~ &\text{ if }  \lceil \Delta/y\rceil > \lceil ( \Delta - x)/y\rceil\\
2 \lceil (\Delta - x)/y\rceil - 1, ~ &\text{ if } \lceil \Delta/y\rceil = \lceil ( \Delta - x)/y\rceil \enspace.
\end{cases}
\]
We will start by proving the following theorem.
\begin{theorem}
	The $x$-maximal $y$-matching problem requires exactly $T_\Delta(x,y)$ rounds for deterministic algorithms in the port numbering model, even on $2$-colored $\Delta$-regular balanced trees.
\end{theorem}
While an upper bound for the port numbering model directly implies also an upper bound for the LOCAL model, the same is not true for lower bounds. However, we will show that these bounds can be lifted to the LOCAL model, and obtain truly tight bounds in the LOCAL model for all combinations of $x$,$y$, and $\Delta$.
Moreover, we will examine how large $\Delta$ can be, as a function of $n$, such that our lower bounds still hold:
\begin{theorem}
	For any $k \ge 1$, and for large enough $\Delta$ and $n$, any randomized algorithm for bipartite $x$-maximal $y$-matching that fails with probability at most $1/n$ in the LOCAL model requires at least $T_\Delta(x,y)$ rounds, unless $T_\Delta(x,y) \ge \frac{1}{3k} \frac{\log\log n}{\log \log \log n}$, or $T_\Delta(x,y) \le \Delta^{1/k}$.
\end{theorem}
Note that this theorem implies tight bounds for problems such as the $\frac{\Delta}{2}$-maximal $\log \Delta$-matching problem, for graphs where $\Delta$ is not too large as a function of $n$.
While a randomized lower bound directly implies the same deterministic lower bound, we will show that we can get \emph{better} deterministic bounds by relaxing our tightness requirements in $\Delta$ from truly tight to asymptotically tight.
\begin{theorem}
	For any $k \ge 1$, and for large enough $\Delta$ and $n$, any deterministic algorithm for bipartite $x$-maximal $y$-matching in the LOCAL model requires $\Omega(\min(T_\Delta(x,y),\frac{1}{k} \log n / \log \log n)$, unless $T_\Delta(x,y) \le \Delta^{1/k}$.
\end{theorem}

\subsection{Related Work}
The maximal matching problem has been widely studied in the literature of distributed computing. We know since the 80s that \mm{} cannot be solved in constant time. In fact, the $\Omega(\log^* n)$ lower bound by Linial for $3$-coloring a cycle \cite{Linial1992} implies a lower bound also for \mm{}. Naor proved that this lower bound holds also for randomized algorithms \cite{Naor1991} .

Concerning upper bounds, we also know since the 80s that \mm{} can be solved in $O(\log n)$ rounds using a randomized algorithm~\cite{Israeli1986}, and there has been a lot of effort in trying to obtain also good deterministic complexities. The first polylogarithmic deterministic algorithm was provided by \citet{Hanckowiak1998}, who showed that \mm{} can be solved in $O(\log^7 n)$ rounds. The same authors later improved the upper bound to $O(\log^4 n)$ \cite{Hanckowiak2001}. More recently, Fischer substantially improved this bound to $O(\log^2 \Delta \log n)$ \cite{fischer17improved}. If the degree of the input graph is small, there is a very efficient algorithm by \citet{panconesi01simple}, that runs in $O(\Delta + \log^* n)$ rounds and thus matches, as a function of $n$, Linial's lower bound. In the meantime, also the upper bound on the randomized complexity of \mm{} has been improved. \citet{Barenboim2012,Barenboim2016} and \citet{fischer17improved} showed that \mm{} can be solved in $\poly \log \log n$ time, by paying only an additive $\log \Delta$ dependency.

In 2004, \citet{Kuhn2004,Kuhn2006,kuhn16local} substantially improved Linial's lower bound: they showed that \mm{} cannot be solved in $o(\sqrt{\log n / \log \log n} + \log \Delta/\log \log \Delta)$, and this bound holds also for randomized algorithms. After that, different works made progress in obtaining a better understanding of the complexity of \mm{} as a function of $\Delta$. First, \citet{Hirvonen2012} proved that if we are only given an edge coloring (but no IDs and no randomness) then indeed this problem requires $\Omega(\Delta)$ rounds. Then, a similar technique has been used to show an $\Omega(\Delta)$ lower bound for fractional matchings, under the assumption that the running time of the algorithm does not depend on $n$ at all \cite{Goos2017}. Finally, a lower bound for the LOCAL model has been proved by \citet{Balliu2019}, who showed that \mm{} cannot be solved in $o(\Delta + \log \log n / \log \log \log n)$ by randomized algorithms and $o(\Delta + \log n / \log \log n)$ by deterministic ones.

\section{Preliminaries}

\subsection{Model of Computing}
In this work, we will mainly consider two different models of distributing computing, namely, the \emph{port numbering} model and the LOCAL model.

In the port numbering model, we are given a graph $G=(V,E)$, where nodes represent computing entities, and edges represent communication links. The computation is synchronous: all nodes start in the same round, and in each round each node can send a different message to each neighbor, receive messages sent by the neighbors, and perform some local computation. In this model, nodes are \emph{anonymous}, that is, they have no IDs, but they can distinguish neighbors using \emph{port numbers}. Each node $v$ has $\deg(v)$ ports, and each edge is connected to a specific port. If an edge $\{u,v\}\in E$ is connected to port $i$ of node $u$ and port $j$ of node $v$, then node $u$ can decide to send a message to its port number $i$, and the message will be received by node $v$ on port $j$. That is, nodes can refer to ports.

Depending on the context, we may assume that nodes initially know the size of the graph $n=|V|$, and the maximum degree of the graph $\Delta$. Also, in a \emph{randomized} algorithm in the port numbering model, each node is provided with an unbounded number of private random bits. We say that a randomized algorithm succeeds with high probability if the output of all nodes is globally correct with probability at least $1-1/n$. The running time of an algorithm is the number of communication rounds required before all nodes output their part of the solution.

The LOCAL model is defined similarly. The only difference is that in this case nodes are not anonymous, that is, they are provided with unique IDs in $\{1,\ldots,n^c\}$ for some constant $c\ge 1$.

\subsection{Automatic Round Elimination}

In order to obtain our lower bounds we will make use of the automatic round elimination framework developed in \cite{Brandt2019}, in the bipartite formulation used first in \cite{Balliu2019}.
For any problem we will consider, the input will be a $\Delta$-regular bipartite graph.
As we will prove lower bounds, this does not restrict the generality of our results.
We refer to the nodes on one side of the bipartition as \emph{white} nodes and the nodes on the other side as \emph{black} nodes.
Each node is aware of the bipartition, i.e., it knows whether it is a white or a black node.

\paragraph{Problems}
In the bipartite round elimination framework (which, for simplicity, we will describe only for $\Delta$-regular graphs), a problem $\Pi$ is formally given by an alphabet $\Sigma$, a \emph{white constraint} $W$, and a \emph{black constraint} $B$.
Both $W$ and $B$ are a collection of words of length $\Delta$ over the alphabet $\Sigma$, where technically each word is to be considered as a multiset, i.e., the order of the $\Delta$ elements in a word does not matter and the same element can appear repeatedly in a word.
A correct output for $\Pi$ is a labeling of the edges of the input graph with one label from $\Sigma$ per edge such that
\begin{enumerate}
	\item the white constraint is \emph{satisfied}, i.e., for each white node $u$, the output labels assigned to the $\Delta$ edges incident to $u$ form a word from $W$, and
	\item the black constraint is \emph{satisfied}, i.e, for each black node $v$, the output labels assigned to the $\Delta$ edges incident to $v$ form a word from $B$.
\end{enumerate}
Each word in $W$ is called a \emph{white configuration}, and each word in $B$ a \emph{black configuration}.
To succinctly represent multiple configurations in one expression, we will make use of regular expressions, e.g., we will write $[M][PO]^{\Delta-1}$ to describe the collection of all configurations consisting of exactly one $M$, and a $P$ or an $O$ at every other position.
For simplicity, we will also call such a regular expression a (white or black) \emph{configuration}.
Where required, we will clarify which kind of configuration is considered by using the terms \emph{single configuration} and \emph{condensed configuration} (the latter indicating a regular expression).
Moreover, we will use the term \emph{disjunction} to refer to parts of a regular expression describing that each choice of a subset of labels is valid, such as $[PO]$. Notice that the encoding of a problem can look quite different if its configurations are condensed in different ways. This is just a syntactic difference: if two different sets of condensed configurations generate the same set of single configurations, then the two sets encode the same constraint.

We remark that even though we only consider $\Delta$-regular graphs, it is possible to also encode locally checkable problems that are not only defined on $\Delta$-regular graphs in a similar way. We also note that if we restrict attention to trees or high-girth graphs, any locally checkable problem can be described in this form. In fact, by increasing the number of labels, it is possible to encode any output constraint that depends on the constant-radius neighborhood of each node. In the remainder of the paper, we will use the term ``locally checkable problem" (or simply ``problem") to refer to problems of the above kind.

\paragraph{Example}
Let us look at an example that shows how to encode \bmm{}. In \bmm{} we basically have to ensure two constraints: a node cannot have two incident edges in the matching, and if a node does not have any incident edge in the matching, then all its neighbors must have at least one.

We start by defining the white constraints as follows:
\begin{align*}
&\M \s \O^{\Delta-1} \\
&\P^{\Delta} \\
\end{align*}
In other words, a white node either outputs $\M$ on an edge (the matched edge) and $\O$ on all the others, or it outputs $\P$ (pointer) on all edges. We now need to ensure that the pointers reach only \emph{matched} black nodes. Thus, we define black constraints as follows:
\begin{align*}
&\M \s [\O\P]^{\Delta-1} \\
&\O^{\Delta} \\
\end{align*}
That is, a black node accepts a pointer only if one of its edges is labeled as $\M$. Clearly, a solution satisfying these constraints is a matching (two $\M$s are never allowed). Maximality is guaranteed by the following observations:
\begin{itemize}
	\item In order for white nodes to not be matched, they need to write $\P$ on all edges. These $\P$s must reach \emph{matched} black nodes, since on the black side $\P$s are accepted only if an $\M$ is present.
	\item In order for black nodes to not be matched, they need to have all edges marked with the label $\O$, and $\O$s are written by white nodes only if they are matched. 
\end{itemize}
Technically, what we defined is not exactly \bmm{}: if we are given a solution for \bmm{}, where the edges are either marked to be part of the matching or not, we do not have edges marked with pointers. Nevertheless, white nodes can produce these pointers in $0$ rounds. That is, the problem that we defined is \emph{equivalent} to \bmm{}.

\paragraph{Algorithms}
We will distinguish between \emph{white algorithms} and \emph{black algorithms}.
In a white algorithm, each white node decides on the output labels for all incident edges whereas black nodes take part in the usual communication but have no say in deciding on the output; in a black algorithm, the roles are reversed.
The \emph{white complexity} (resp.\ \emph{black complexity}) of a given problem is the usual time complexity of the problem where we restrict attention to white algorithms (resp.\ black algorithms).
White and black complexities cannot differ by more than one round as any white node can inform any black neighbor of the intended output label for the connecting edge, and vice versa.

Notice that if we consider, e.g., a white algorithm, black nodes do not actually need to know the output given by white nodes. If we consider the more standard assumption where \emph{both} nodes that are incident to the same edge know the output for that edge, we see that such an algorithm requires at most one round more than what is required by either a white or a black algorithm. However, in the bipartite round elimination framework, such algorithms require an extra step of argumentation which we omit for simplicity, by considering only white and black algorithms. We emphasize that the \emph{tightness} of our \mm{} bound does not depend on this choice, just the bound itself: in the setting where both endpoints of an edge have to know the output for that edge, the tight bound is $2 \Delta$, instead of $2 \Delta - 1$.

\paragraph{Round Elimination Theorems}
The automatic round elimination theorem given in \cite[Theorem 4.3]{Brandt2019}, roughly speaking, states that for any locally checkable problem $\Pi$, there exists another locally checkable problem $\Pi'$ that can be solved exactly one round faster if we restrict attention to high-girth graphs, i.e., graphs where the cycle of smallest length is sufficiently long.
A useful fact observed in \cite{Balliu2019} is that the given proof also extends to the case of hypergraphs, and hence can also be phrased in the context of bipartite graphs, by interpreting nodes as one side of the bipartition and hyperedges as the other side.
In order to satisfy the conditions of \cite[Theorem 4.3]{Brandt2019}, we will restrict our attention to the class of regular, bipartite graphs with a girth of at least $4 \Delta + 2$, for the remainder of the paper.
We will see that our lower bounds hold already for this restricted graph class.
Now, we can formulate \cite[Theorem 4.3]{Brandt2019} in our setting as follows.\footnote{For the reader interested in the technical subtleties of our rephrasing, four remarks are in order: 1) The edge orientations prescribed in \cite[Theorem 4.3]{Brandt2019} are simply given by the port numberings (of the nodes on one side of the bipartition) in our setting. 2) In our setting the nodes do not see the port numbering of adjacent nodes in a $0$-round algorithm, while the nodes in \cite{Brandt2019} do see the edge orientations; however, it is straightforward to check that the proof is not affected by this change. 3) If $\Pi$ can be solved in $0$ rounds, then the proof of \cite[Theorem 4.3]{Brandt2019} ensures that this also holds for $\Pi'$; hence we can replace the condition of a strictly positive complexity of $\Pi$ by a minimum expression. 4) As the nodes on one side of the bipartition in our setting correspond to (hyper)edges in the original setting, the step from $\Pi$ to $\REB(\Pi)$ (resp.\ $\REW(\Pi)$) in our setting corresponds either to the first step from $\Pi$ to the intermediate problem $\Pi_{1/2}$ given in the proof of \cite[Theorem 4.3]{Brandt2019}, or to the step from $\Pi_{1/2}$ to the final problem $\Pi_1$ (depending on whether we consider white or black nodes as (hyper)edges).}

\begin{theorem}[\cite{Brandt2019}, rephrased]\label{thm:speeduppaper}
	Let $\Pi$ be a locally checkable problem with white (resp.\ black) complexity $T(n, \Delta)$. Then there exists a locally checkable problem $\REB(\Pi)$ (resp.\ $\REW(\Pi)$) with black (resp.\ white) complexity $\min \{ 0, T(n, \Delta) - 1 \}$.
\end{theorem}

The problem $\REB(\Pi)$ is constructed explicitly in \cite{Brandt2019}.
Translated to our setting, we obtain $\REB(\Pi)$ from $\Pi$ as follows.

Let $\Sigma_\Pi$, $W_\Pi$, and $B_\Pi$ denote the alphabet, white constraint, and black constraint of $\Pi$.
The alphabet $\Sigma_{\REB(\Pi)}$ of $\REB(\Pi)$ is simply the set $2^{\Sigma_\Pi}$ of all subsets of $\Sigma_\Pi$.
In order to describe the black constraint of $\REB(\Pi)$, we first construct an intermediate collection $B'$ of configurations over $\Sigma_{\REB(\Pi)}$.
Let $B'$ be the collection of all configurations $F_1, \dots, F_{\Delta}$ with $F_i \in 2^{\Sigma_\Pi}$ such that \emph{for each} choice $f_1 \in F_1, \dots, f_\Delta \in F_\Delta$ of labels from $\Sigma_\Pi$, it holds that $f_1, \dots, f_\Delta$ is a configuration in $B_\Pi$.
We now obtain $B_{\REB(\Pi)}$ from $B'$ by removing all configurations $F_1, \dots, F_{\Delta}$ that are not \emph{maximal}, i.e., for which it is possible to obtain another configuration from $B'$ by adding elements to the $F_i$ (more precisely, at least one element to at least one of the $F_i$).
Since the above removal process ensures that for each non-maximal configuration there always remains a ``super-configuration" in the collection, the order in which the non-maximal configurations are removed does not matter.

Similarly, to obtain $W_{\REB(\Pi)}$, we first construct an intermediate collection $W'$ of configurations over $\Sigma_{\REB(\Pi)}$.
Let $W'$ be the collection of all configurations $F'_1, \dots, F'_{\Delta}$ with $F'_i \in 2^{\Sigma_\Pi}$ such that \emph{there exists} a choice $f'_1 \in F'_1, \dots, f'_\Delta \in F'_\Delta$ of labels from $\Sigma_\Pi$ such that $f'_1, \dots, f'_\Delta$ is a configuration in $W_\Pi$.
We now obtain $W_{\REB(\Pi)}$ from $W'$ by removing each configuration that contains some set $F'_i$ that does not appear in any of the configurations in the black constraint $B_{\REB(\Pi)}$.

So, roughly speaking, apart from some simplifications on top of it, we obtain the new black constraint by ``applying" the universal quantifier to the old black constraint, and the new white constraint by ``applying" the existential quantifier to the old white constraint.

We define $\REW(\Pi)$ analogously to $\REB(\Pi)$, with the only difference that the role of white and black is reversed.
The following observation follows immediately from the definition of our functions $\REB()$ and $\REW()$.
\begin{observation}\label{obs:iterate}
	Let $\Pi$ be some problem, and assume we have already computed the black constraint of $\REB(\Pi)$ (resp.\ the white constraint of $\REW(\Pi)$).
	Then the white constraint of $\REB(\Pi)$ can be obtained by iterating through the white configurations of $\Pi$ and replacing in each configuration each label $L$ by the disjunction of all sets that occur in the black constraint of $\REB(\Pi)$ and contain label $L$.
	Similarly, the black constraint of $\REW(\Pi)$ can be obtained by iterating through the black configurations of $\Pi$ and replacing in each configuration each label $L$ by the disjunction of all sets that occur in the white constraint of $\REW(\Pi)$ and contain label $L$.
\end{observation}

As the alphabet of a problem obtained by applying the function $\REB()$ (resp.\ $\REW()$) consists of sets of the original labels, we will need to be careful with notation. In order to clearly distinguish the set consisting of some labels $\X, \Y, \Z$ from the disjunction $[\X \Y \Z]$, we will write it as $\bXYZ$.
Moreover, we say that a configuration $F_1, \dots, F_{\Delta}$ consisting of sets of labels \emph{can be extended} to a configuration $\mathcal C$ if $\mathcal C$ can be obtained from  $F_1, \dots, F_{\Delta}$ by adding (potentially everywhere $0$) elements to the sets $F_i$.

\paragraph{Example}
We will now show an example of the application of this technique. We will consider the bipartite sinkless orientation problem \cite{Brandt2016}, where white constraints can be simpliy described as $\B \s [\A\B]^{\Delta-1}$, and black constraint can be described as $\A \s [\A\B]^{\Delta-1}$. Intuitively, the label $\A$ represents an edge oriented from black to white, while $\B$ represents an edge oriented from white to black, and the constraints require both black and white nodes to have at least an outgoing edge: $\B$ for white nodes and $\A$ for black nodes.

We start by applying the universal quantifier on the black constraint. Essentially, we must forbid all words where all the $\Delta$ sets contain a $\B$, otherwise it would be possible to pick the configuration $\B^\Delta$ that is not allowed by the black constraint. Hence, we obtain the following:
	\begin{align*}
&\bA \s [\bA\bB\bAB]^{\Delta-1}
\end{align*}
We can now apply maximality, and since in the disjunction $[\bA\bB\bAB]$ the label $\bAB$ strictly contains all the others, we get the following:
	\begin{align*}
&\bA \s \bAB^{\Delta-1}
\end{align*}
We can now apply the existential quantifier on the white constraint. The white constraint basically requires to be able to pick at least one $\B$, hence, we now need at least one set containing a $\B$:
\begin{align*}
&[\bB\bAB] \s [\bA\bB\bAB]^{\Delta-1}
\end{align*}
By removing labels that are not used in the universal step, we get the following:
\begin{align*}
&\bAB  \s [\bA\bAB]^{\Delta-1}
\end{align*}
We can now rename the sets to obtain something more readable. We can use the following renaming:
\begin{align*}
	\bA & \rightarrow \A \\
	\bAB & \rightarrow \B
\end{align*}
Thus, the black constraint of $\REB(\Pi)$ can be described as $\A \s \B^{\Delta-1}$, while the white constraint can be described as $\B \s [\A\B]^{\Delta-1}$.

\paragraph{Relaxations}
As mentioned in the introduction, a crucial part of successfully applying automatic round elimination is to find good relaxations of problems.
We say that a problem $\Pi'$ is a \emph{white relaxation} of a problem $\Pi$ if there is a $0$-round white algorithm that transforms any arbitrary correct output for $\Pi$ into a correct output for $\Pi'$.
More specifically, in this $0$-round algorithm, each white node sees, for all incident edges $e$, the given output label (for $\Pi$) on $e$, and its task is to relabel each incident edge (possibly with the same label as before) such that the global output is correct.
We define a \emph{black relaxation} analogously.
It immediately follows from the definition that the relaxation of a relaxation of a problem is again a relaxation of the problem.

A simple way to find a relaxation of a problem is to replace a fixed label everywhere by another label.
Any white or black node can solve the new problem in $0$ rounds given a solution to the old problem, by just performing the corresponding relabeling everywhere.
\begin{observation}\label{obs:replace}
	Let $\Pi$ be a problem, and $K, L \in \Sigma_\Pi$ labels.
	Then replacing all occurrences of label $K$ in both $W_\Pi$ and $B_\Pi$ by label $L$ results in a problem $\Pi'$ that is a relaxation of $\Pi$.
\end{observation}

A more interesting way to find relaxations of problems given via white and black constraints starts by ordering the labels occuring in the constraints according to their \emph{strength}, which roughly corresponds to their usefulness for outputting a correct configuration.
Consider the black constraint $B_\Pi$ of some problem $\Pi$.
For two labels $K, L \in \Sigma_\Pi$, we say that $L$ is \emph{at least as strong} as $K$ (according to $B_\Pi$) if for each $K$ appearing in some configuration in $B_\Pi$, replacing that $K$ by an $L$ results again in some configuration in $B_\Pi$.  
Equivalently, we say that $K$ is \emph{at least as weak} as $L$ (according to $B_\Pi$), and write $K \leq_B L$.
If $K$ is at least as weak as $L$, but $L$ is not at least as weak as $K$, then we say that $L$ is \emph{stronger} than $K$, and $K$ \emph{weaker} than $L$.
If $K$ is at least as weak as $L$, and $L$ is at least as weak as $K$, we say that $K$ and $L$ are equally strong.
We define these concepts and notations analogously for white constraints.
Now, we can use the strengths of labels to find relaxations as follows.

\begin{lemma}\label{lem:merge}
	Let $\Pi$ be a problem, and $K, L \in \Sigma_\Pi$ labels such that $L$ is at least as strong as $K$ according to $B_\Pi$ (resp.\ $W_\Pi$).
	Then replacing an arbitrary number of labels $K$ in $W_\Pi$ (resp.\ $B_\Pi$) by label $L$ results in a problem $\Pi'$ that is a white (resp.\ black) relaxation of $\Pi$.
\end{lemma}
\begin{proof}
	For reasons of symmetry it is sufficient to prove the lemma for the case of replacing the labels in $W_\Pi$.
	We obtain a $0$-round algorithm as required in the definition of a relaxation as follows.
	Given a solution to $\Pi$ (on each incident edge), each white node $v$ simply replaces as many incident occurrences of label $K$ by $L$ as have been replaced in the configuration the solution of $\Pi$ around $v$ corresponds to.
	This satisfies the white constraint by definition.
	But also the black constraint is satisfied since replacing occurrences of label $K$ by $L$ preserves that a configuration is contained in $B_\Pi$ due to the fact that $L$ is at least as strong as $K$, and $B_{\Pi'} = B_\Pi$.
\end{proof}

Recall that the labels used to describe a problem obtained by applying the function $\REB()$ (resp.\ $\REW()$) to some problem $\Pi$ are sets of labels of $\Pi$.
The following observation follows immediately from the definition of $\REB()$ (resp.\ $\REW()$).
\begin{observation}\label{obs:strongerset}
	Let $\Pi$ be a problem, and $K, L \in 2^{\Sigma_\Pi}$ labels of the problem $\REB(\Pi)$ (resp.\ $\REW(\Pi)$).
	If $K \subseteq L$, then $L$ is at least as strong as $K$ according to the white constraint of $\REB(\Pi)$ (resp.\ to the black constraint of $\REW(\Pi)$).
\end{observation}

To visualize the strengths of the labels according to the black (resp.\ white) constraint of a problem, we can draw a diagram as follows.
For any two labels $K, L \in \Sigma_\Pi$, we draw an arrow from $K$ to $L$ if
\begin{enumerate}
	\item $K$ and $L$ are equally strong, or
	\item $L$ is stronger than $K$ and there is no label $M$ such that $M$ is stronger than $K$, and $L$ is stronger than $M$.
\end{enumerate}
We call the obtained diagram the \emph{black (resp.\ white) diagram} of $\Pi$.

\paragraph{Example}
Consider the following problem $\Pi$ (for the curious reader, this problem is exactly $2$ rounds easier than \bmm{}). The white constraint is the following:
\begin{align*}
& \M \s \O^{\Delta-1}\\
& \Y \s \P^{\Delta-1}\\
& \X \s \Z \s \O^{\Delta-2}
\end{align*}
The black constraint is the following:
\begin{align*}
& [\M\Y\X] \s [\P\Y\O\X]^{\Delta-1}\\
& [\Z\M\P\Y\O\X] \s [\O\X]^{\Delta-1}
\end{align*}
The black diagram of $\Pi$ is the following:
\begin{center}
	\includegraphics[width=0.2\textwidth]{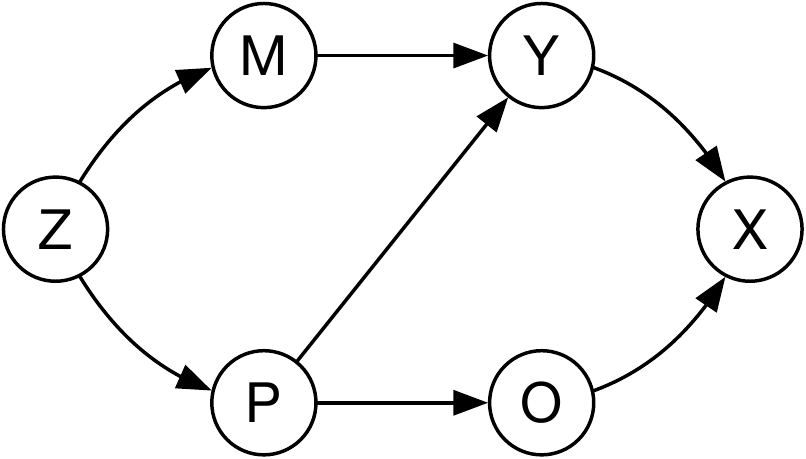}
\end{center}
Notice that there is an arrow from $\Y$ to $\X$, since each time $\Y$ is allowed in a black configuration, the label $\X$ is allowed as well. We can thus simplify the problem $\Pi$ according to Lemma \ref{lem:merge} using labels $\X$ and $\Y$. In particular, we can replace \emph{all} occurences of $\Y$ with $\X$ in the white constraint, and thus get rid of $\Y$ also in the black constraint. We obtain the following new white constraint:
\begin{align*}
& \M \s \O^{\Delta-1}\\
& \X \s \P^{\Delta-1}\\
& \X \s \Z \s \O^{\Delta-2}
\end{align*}
The new black constraint is the following:
\begin{align*}
& [\M\X] \s [\P\O\X]^{\Delta-1}\\
& [\Z\M\P\O\X] \s [\O\X]^{\Delta-1}
\end{align*}

\paragraph{White-Black Dualism}
So far, we have seen a number of definitions and results that have two versions: one for white nodes, algorithms, configurations, etc., and one for the black equivalent.
In fact, the white and black versions are completely dual: they only differ in exchanging the role of white and black.
Moreover, this behavior will hold throughout the entirety of the paper.
Hence, we will use the following convention in the remainder of the paper.
\begin{convention}
	For simplicity, we will formulate any theorem and lemma for which there is a white and a black version only in one of the two versions.
	Moreover, by giving a theorem and lemma for which all ingredients are defined if we exchange the terms ``white" and ``black", we implicitly state that also its dual version holds.
	We will refer to the dual of a stated theorem or lemma by simply referring to the theorem or lemma in its original version; the context in which the statement is referred indicates which version is meant.
\end{convention}

\paragraph{From the port numbering model to the LOCAL model}
The round elimination theorem can be used to get lower bounds for deterministic algorithms in the port numbering model, where it is assumed that nodes have no IDs and no access to random bits. The usual way \cite{Balliu2019, binary} to lift lower bounds obtained with this technique to the LOCAL model is the following. First, we incorporate the analysis of failure probabilities in the round elimination theorem, i.e., we show that if $\Pi$ can be solved in $T$ rounds using a white randomized algorithm \emph{with some failure probability $p$}, then $\REB(\Pi)$ can be solved in $T-1$ rounds using a black randomized algorithm \emph{with some failure probability $p'$}, where $p'$ is not much larger than $p$. Using this version of the theorem, we can repeatedly apply round elimination, until either we get a $0$ round solvable problem or we get a too large failure probability. This allows us to obtain a lower bound for randomized algorithms for the port numbering model, and since by using randomness it is also possible to generate unique IDs with high probability of success, then the same lower bound holds in the LOCAL model as well. 

Notice that we do not have to actually prove a randomized version of Theorem \ref{thm:speeduppaper}: a randomized version of the round elimination theorem has been shown in \cite{binary}, where it has also been shown that, if the number of labels is bounded at each step, then we can automatically lift a lower bound obtained with this technique to the LOCAL model.

Also, since a lower bound for randomized algorithms implies also a lower bound for deterministic algorithms, we immediately get as a corollary a lower bound for the LOCAL model for deterministic algorithms. 

We can then get even better lower bounds for deterministic algorithms by exploiting known \emph{gap} results: we know that for locally checkable problems some complexities are not possible, and if we get as lower bound a complexity $T$ that falls into one of these gaps, we immediately get as new lower bound the smallest complexity $T'$, larger than $T$, for which the gap does not hold anymore.

\section{Roadmap}
In the breakthrough result of \citet{Balliu2019}, the round elimination technique yielded only an $\Omega(\sqrt{\Delta})$ lower bound for \bmm{}, and different techniques were necessary to lift the result to an $\Omega(\Delta)$ lower bound. One of our contributions is to show the reason why a full $\Omega(\Delta)$ lower bound was not achieved via round elimination: the relaxations performed at each step were too severe. In fact, the main issue of the approach in \cite{Balliu2019} is the following. After each step of round elimination, some simplifying relaxations are performed in order to obtain a problem that can be described using just $4$ labels. By performing such simplifications, a special label called \emph{wildcard} appears: this is a powerful label that can replace any other label in any configuration without invalidating the configuration. The issue is now the following: by performing a round elimination step on a problem where valid configurations contain some number of wildcards, the obtained resulting problem contains configurations with \emph{many more} wildcards. In particular, the number of wildcards grows \emph{quadratically}. Hence, after $O(\sqrt{\Delta})$ round elimination steps the obtained problem description contains so many wildcards that the problem is $0$-round solvable.

What we can show is that, if we allow just one more label at each step, hence $5$ instead of $4$, then the number of wildcards does not grow at all; instead we obtain a \emph{linear} growth on some parameter that controls how easy the problem is. In this way, for \bmm{}, we can perform $2\Delta-1$ steps of round elimination before getting to a $0$-rounds solvable problem. Since we can also provide an \emph{upper bound} with the same number of rounds, this implies that the simplifications that we perform are not making the problem easier at all: they produce a problem that is easier to describe but that has the same round complexity as the one before the respective simplification. We will actually prove a stronger result: using $5$ labels we can prove \emph{exact} bounds for the \emph{whole} family of bipartite $x$-maximal $y$-matchings.

Hence, the family of problems that we provide in the lower bound proof really captures the \emph{essence} of \bmm{} and its variants: for each given variant of \bmm{} and any $i$, we describe \emph{in a compact form} a problem that is \emph{exactly} $i$ rounds easier than the given variant.

\paragraph{Round elimination} We start by defining a family of problems $\Psi_\Delta(a,b,c)$, described by $3$ parameters. Then, we prove that the problem $\Psi_\Delta(y,x,0)$ is exactly $2$ rounds easier than the bipartite $x$-maximal $y$-matching problem (recall that bipartite $0$-maximal $1$-matching is standard \bmm{}). Then, we relate problems in the family: we will show that the problem  $\Psi_\Delta(a,c+a,b)$ is at least one round easier than $\Psi_\Delta(a,b,c)$ (the results that we provide in the upper bound section will imply that $\Psi_\Delta(a,c+a,b)$ is actually \emph{exactly} one round easier than $\Psi_\Delta(a,b,c)$). In this way, we get a full characterization of all the problems in the family. Crucially, all the problems of the family are described using only $5$ labels, and while the result of the round elimination technique may contain more than $5$ labels, we will provide relaxations that allow us to map these problems back to this family.

\paragraph{Lower bounds for bipartite $x$-maximal $y$-matchings}
We will then prove that, for some values of $a,b,c$, the problem $\Psi_\Delta(a,b,c)$ cannot be solved in $0$ rounds, and we will then show what lower bounds this implies for bipartite $x$-maximal $y$-matchings. In particular, we will obtain tight-in-$\Delta$ lower bounds for the \emph{port numbering} model.

\paragraph{Upper bounds}
Then, we will prove upper bounds for the whole family of bipartite $x$-maximal $y$-matchings. These upper bounds will match exactly the lower bounds that we provided.

\paragraph{Lifting the bounds to the LOCAL model}
At this point we have lower bounds for plenty of variants of matchings, for the port numbering model. We now need to lift these bounds to the LOCAL model. The first step is to prove a randomized lower bound for the port numbering model, that directly implies a lower bound for the LOCAL model as well, since it is possible to generate unique IDs in constant time. Then, while a randomized lower bound directly implies a deterministic lower bound, we will obtain a \emph{better} deterministic lower bound using standard techniques.

\paragraph{Behind the scenes}
We will then informally discuss about what we mean by automatic lower and upper bounds. We will briefly explain how, part of our results concerning both lower and upper bounds can be actually obtained automatically, not in theory but also in practice.

\paragraph{Open questions}
We will finally conclude with some open questions regarding round elimination, and in particular about how different can be lower bounds obtained by using bounded automatic round elimination, compared to what can be obtained by using the standard version of round elimination for the same problem. 

\section{Lower Bounds}\label{sec:lower}
In this section, we will prove a lower bound for \bmm{} that we will show to be tight in the LOCAL model in Sections \ref{sec:upper} and \ref{sec:lifting}. In particular, our lower bound holds even in the restricted setting of $\Delta$-regular trees.
In order to obtain this bound, we will define a family $\mathcal F$ of problems that will help us to describe the behavior of \bmm{} in the round elimination framework. 
As it turns out, our way of proving a lower bound via such a family of problems is very \emph{robust}: we can also obtain tight lower bounds for very natural variants of \bmm{} by generalizing our problem family.
In fact, we believe that the general lower bound idea should also work for many other variants of \bmm{}; however we will only give an explicit proof for the variants we consider to be the most natural extensions of \bmm{}, namely those obtained by relaxing the \emph{packing} and/or \emph{covering} constraint used to define \bmm{}.
To be more precise, if we relax the packing constraint, we allow each node to be matched with up to $y$ many neighbors, for some parameter $y$, and if we relax the covering constraint, we allow each unmatched node to have up to $x$ neighbors that are unmatched themselves, giving rise to the already defined notion of $x$-maximal $y$-matchings.
\bmm{} appears as a special case in this family, where we set $x = 0$ and $y = 1$; hence we will only give a general proof for the whole family and obtain the proof for \bmm{} as a special case.

In order to obtain the desired lower bounds for all problems in this family, we define the aforementioned problem family $\mathcal F$ (or more precisely, its extension that also describes the behavior of all bipartite $x$-maximal $y$-matching problems), parameterized (for each fixed $\Delta$) by three values $a, b, c$. 
For each $1 \leq a \leq \Delta-1$ and $0 \leq b, c \leq \Delta$ such that $a + b \leq \Delta$ and $a + c \leq \Delta$, denote by $\Psi_\Delta^W(a,b,c)$ (resp.\ $\Psi_\Delta^B(a,b,c)$) the problem given by the white (resp.\ black) configurations
\begin{align*}
	&\X^{a-1} \s \M \s \O^{\Delta-a} \\
	&\X^a \s \O^{b} \s \P^{\Delta-a-b} \\
	&\X^a \s \Z \s \O^{\Delta-a-1}
\end{align*}
and the black (resp.\ white) configurations
\begin{align*}
	&[\M \Z \P \O \X]^{a-1} \s [\M \X] \s [\P \O \X]^{\Delta-a} \\
	&[\M \Z \P \O \X]^a \s [\P \O \X]^{c} \s [\O \X]^{\Delta-a-c} \\
	&[\M \Z \P \O \X]^a \s [\X] \s [\P \O \X]^{\Delta-a-1} \enspace.
\end{align*}
The black diagram of $\Psi_\Delta^W(a,b,c)$ is represented in Figure \ref{fig:diagrampsi}. For extreme values of $a,b,c$ additional arrows may be present.

\begin{figure}
	\centering
	\includegraphics[width=0.2\textwidth]{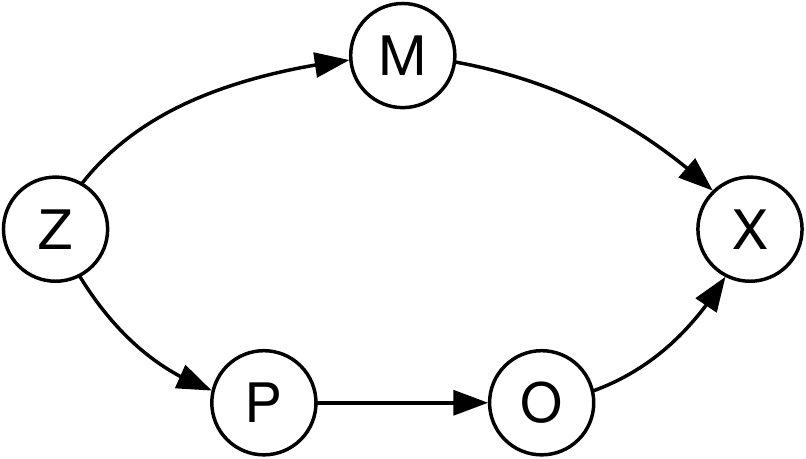}
	\caption{The black diagram of $\Psi_\Delta^W(a,b,c)$.}
	\label{fig:diagrampsi}
\end{figure}

We now proceed as follows. In Lemma \ref{lem:varinitone} and Lemma \ref{lem:varinittwo} we will prove that $\Psi_\Delta^W(x,y,0)$ is exactly two rounds easier than the bipartite $x$-maximal $y$-matching problem. In Lemma \ref{lem:0round} we will prove that for most parameter values, the problems in the family $\Psi_\Delta^W(a,b,c)$ are not $0$-round solvable. Finally, in Theorem \ref{thm:pnlower} we will combine all the results of this section to prove a lower bound for the bipartite $x$-maximal $y$-matching problem. We will now start by relating problems in the family that we just defined, by showing in Lemma \ref{lem:vartech} that by applying Theorem \ref{thm:speeduppaper} on $\Psi^W_{\Delta}(a,b,c)$ and performing some relaxations, we obtain a problem that is still in the same family, but with different parameters.

\begin{lemma}\label{lem:vartech}
	For any $\Delta \geq 2$, $1 \leq a \leq \Delta - 1$, $0 \leq b \leq \Delta-a$, and $0 \leq c \leq \Delta-a$, problem $\Psi^B_{\Delta}(a,d,b)$ is a black relaxation of $\REB(\Psi^W_{\Delta}(a,b,c))$, where $d = \min \{ c + a, \Delta - a \}$.
\end{lemma}

\begin{proof}
	Recall the definition of extending a configuration, and let $d = \min \{ c + a, \Delta - a \}$.
	We start by showing that any configuration in the black constraint of $\REB(\Psi^W_{\Delta}(a,b,c))$ can be extended to one of the configurations
	\begin{align*}
		&\bMZPOX^{a-1} \s \bMX \s \bPOX^{\Delta-a} \\
		&\bMZPOX^a \s \bPOX^{d} \s \bOX^{\Delta-a-d} \\
		&\bMZPOX^a \s \bX \s \bPOX^{\Delta-a-1} \enspace.
	\end{align*}
	Consider an arbitrary configuration $\mathcal C$ in the black constraint of $\REB(\Psi^W_{\Delta}(a,b,c))$, and recall that the black constraint of $\REB(\Psi^W_{\Delta}(a,b,c))$ is obtained by applying the universal quantifier to the white constraint of $\Psi^W_{\Delta}(a,b,c)$.
	We distinguish three cases.
	
	If $\mathcal C$ contains the set $\bX$, then at most $a$ of the other $\Delta - 1$ sets in $\mathcal C$ can contain an $\M$ or $\Z$ since in each black configuration of $\Psi^W_{\Delta}(a,b,c)$ there are at most $a$ labels from $\{ M, Z \}$.
	Hence the remaining $\Delta-a-1$ labels must be subsets of $\bPOX$ each.
	It follows that $\mathcal C$ can be extended to $\bMZPOX^a \s \bX \s \bPOX^{\Delta-a-1}$.

	If $\mathcal C$ contains the set $\bMX$, then at most $a-1$ of the other $\Delta - 1$ sets in $\mathcal C$ can contain an $\M$ or $\Z$ since otherwise we would again be able to choose $a+1$ labels contained in $\{ M, Z \}$ from $a+1$ sets in $\mathcal C$ which cannot result in a black configuration of $\Psi^W_{\Delta}(a,b,c)$, no matter which labels are picked from the remaining $\Delta-a-1$ sets.
	Hence the remaining $\Delta-a$ labels must be subsets of $\bPOX$ each, and similarly to the previous case, we see that $\mathcal C$ can be extended to $\bMZPOX^{a-1} \s \bMX \s \bPOX^{\Delta-a}$.

	Consider the last remaining case, i.e., that $\mathcal C$ contains neither the set $\bX$ nor the set $\bMX$.
	Since $\X$ is a label that is at least as strong as $\M$ according to the black constraint of $\Psi^W_{\Delta}(a,b,c)$, any set in $\mathcal C$ that contains $\M$ must also contain $\X$ as otherwise adding $\X$ to that particular set would still result in a black configuration of $\REB(\Psi^W_{\Delta}(a,b,c))$ which would violate the maximality condition in the definition of $\REB()$.
	It follows that $\mathcal C$ does not contain the set $\bM$.
	By a similar argument to the previous one, any set in $\mathcal C$ containing $\Z$ must also contain $\P$, and any set containing $\P$ must also contain $\O$, due to $\P$ being a label as strong as $\Z$ and $\O$ being a label as strong as $\P$.
	Hence, it follows that each set in $\mathcal C$ must contain the label $\O$, since if a set does not contain $\O$ then it must not contain $\P$ and $\Z$ as well, and this in turn would imply that the set is either $\bX$ or $\bMX$, and this case has already been covered.
	This implies that at most $a+c$ sets in $\mathcal C$ can contain the label $\P$ as otherwise we could choose $a+c+1$ times the label $\P$ and $\Delta-a-c-1$ times the label $\O$ from the sets in $\mathcal C$ which does not yield a black configuration of $\Psi^W_{\Delta}(a,b,c)$.
	Moreover, with an argumentation analogous to the one in the first case, we see that at most $a$ sets in $\mathcal C$ can contain an $\M$ or $\Z$, and we can extend these sets to $\bMZPOX$.
     The other $\Delta-a$ sets must be subsets of $\bPOX$, and at most $d = \min \{ c + a, \Delta - a \}$ of them can contain a $\P$. We extend these sets to $\bPOX$. All other sets are extended to $\bOX$.
	It follows that $\mathcal C$ can be extended to $\bMZPOX^a \s \bPOX^{d} \s \bOX^{\Delta-a-d}$.
	
	Now, Lemma \ref{lem:merge} in conjunction with Observation \ref{obs:strongerset} tells us that replacing the black constraint of $\REB(\Psi^W_{\Delta}(a,b,c))$ by the configurations
	\begin{align*}
		&\bMZPOX^{a-1} \s \bMX \s \bPOX^{\Delta-a} \\
		&\bMZPOX^a \s \bPOX^{d} \s \bOX^{\Delta-a-d} \\
		&\bMZPOX^a \s \bX \s \bPOX^{\Delta-a-1} \enspace.
	\end{align*}
	will result in a black relaxation of $\REB(\Psi^W_{\Delta}(a,b,c))$.
	What is left to be done is to compute the white constraint of this relaxation.
	As we performed the above relaxation, formally we cannot directly apply Observation \ref{obs:iterate}, but the same idea works: the definition of our function $\REB()$ still ensures that we can obtain the white constraint of the relaxation by iterating through the white configurations of $\Psi^W_{\Delta}(a,b,c)$ and replacing in each configuration each label $L$ by the disjunction of all sets that occur in the black constraint of the relaxation and contain label $L$.
	We obtain that the white constraint is given by the configurations
	\begin{align*}
		&[\bX \bMX \bOX \bPOX \bMZPOX]^{a-1} \s [\bMX \bMZPOX] \s [\bOX \bPOX \bMZPOX]^{\Delta-a} \\
		&[\bX \bMX \bOX \bPOX \bMZPOX]^a \s [\bOX \bPOX \bMZPOX]^{b} \s [\bPOX \bMZPOX]^{\Delta-a-b} \\
		&[\bX \bMX \bOX \bPOX \bMZPOX]^a \s [\bMZPOX] \s [\bOX \bPOX \bMZPOX]^{\Delta-a-1} \enspace.
	\end{align*}
	Finally, renaming the sets in the black and white constraint of the relaxation according to
	\begin{align*}
		\bX &\rightarrow \Z \\
		\bOX &\rightarrow \P \\
		\bPOX &\rightarrow \O \\
		\bMX &\rightarrow \M \\
		\bMZPOX &\rightarrow \X \enspace,
	\end{align*}
	shows that the relaxation is identical to $\Psi^B_{\Delta}(a,d,b)$, thereby proving the lemma.
\end{proof}

In order to relate the $x$-maximal $y$-matching problem with the family $\mathcal{F}$, we now redefine the bipartite $x$-maximal $y$-matching problem in a way that conforms to the formalism used in Theorem \ref{thm:speeduppaper}. We will then prove that by applying Theorem \ref{thm:speeduppaper} twice on this encoded version of $x$-maximal $y$-matching and performing some relaxations, we get some problem in the family $\mathcal{F}$. We will use $4$ labels: $\M$, $\O$, $\X$, and $\P$. The label $\M$ represents a matched edge, and for both black and white nodes it can appear on at most $y$ incident edges. Also, we want unmatched white nodes to prove that they have enough matched neighbors. Thus, we require that they output at least $\Delta-x$ pointers using the label $\P$. Then, they can output $\X$ on all the other incident edges, where the label $\X$ represent a wildcard. The $\P$ pointers must be incident to only black matched nodes. In order to require that also unmatched black nodes have enough matched white neighbors, we limit the number of labels $\X$ incident to an unmatched black node to $x$. The label $\O$ represents an unmatched edge, where no pointer nor wildcard has been used. More formally, we define the bipartite $x$-maximal $y$-matching problem as follows.

For any $0 \leq x \leq \Delta$ and $1 \leq y \leq \Delta$, denote by $\Phi_\Delta^W(x,y)$ (resp.\ $\Phi_\Delta^B(x,y)$) the problem given by the white (resp.\ black) configurations
\begin{align*}
	&[\M \O \X]^{y-1} \s [\M] \s [\O \X]^{\Delta-y} \\
	&[\P \X]^{x} \s [\P]^{\Delta-x} \enspace,
\end{align*}
and the black (resp.\ white) configurations
\begin{align*}
	&[\M \P \O \X]^{y-1} \s [\M] \s [\P \O \X]^{\Delta-y} \\
	&[\O \X]^{x} \s [\O]^{\Delta-x} \enspace.
\end{align*}
The black diagram of $\Phi_\Delta^W(x,y)$ is represented in Figure \ref{fig:diagramphi1}. For extreme values of $x$ and $y$ additional arrows may be present.

\begin{figure}
	\centering
	\includegraphics[width=0.2\textwidth]{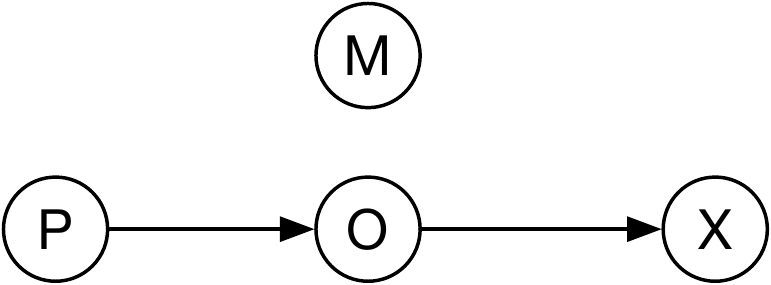}
	\caption{The black diagram of $\Phi_\Delta^W(x,y)$.}
	\label{fig:diagramphi1}
\end{figure}

We now argue that the problem that we just defined \emph{is} the $x$-maximal $y$-matching problem previously defined in the introduction. In particular, this definition correctly encodes the bipartite $x$-maximal $y$-matching problem, where the solution is given by edges labeled $\M$. In fact, first, note that each node is incident to at most $y$ edges labeled $\M$, hence the packing constraint is not violated. Then, note that unmatched white nodes are incident to at least $\Delta-x$ $\P$s, and since black nodes incident to a $\P$ must also be incident to an $\M$, then each white node has at most $x$ unmatched neighbors. Finally, note that unmatched black nodes are incident to at least $\Delta-x$ $\O$s, and since white nodes incident to an $\O$ must also be incident to an $\M$, then each black node has at most $x$ unmatched neighbors. Thus, the covering constraint is not violated. We now need to show that given a solution to the bipartite $x$-maximal $y$-matching problem we can output a solution for this problem with the encoding that we defined. This can be performed in $1$ round of communication, required by white nodes to be aware of which black neighbors are actually matched. Matched white nodes output $\M$ on all matched edges and $\O$ on all the others, while unmatched white nodes output $\P$ on all edges connecting them to matched black nodes, and $X$ on all the others. This is a valid solution, since the number of $\M$s incident to each node is at most $y$, at most $x$ neighbors of unmatched white nodes are also unmatched and thus white nodes are incident to at most $x$ $\X$s, and at most $x$ neighbors of unmatched black nodes are also unmatched and thus at most $x$ white nodes output $\X$ on their ports. Notice that the bipartite $x$-maximal $y$-matching problem now has two different meanings:
\begin{itemize}
	\item The natural and intuitive version where both black and white nodes know which edges are part of the matching, and nothing else.
	\item The locally checkable encoded version, where only white nodes need to know the output, but they need to know which edges are part of the matching \emph{and} which edges must be marked with $\P$.
\end{itemize}
All the results that we prove in the paper about $x$-maximal $y$-matchings will refer to the locally checkable encoded version. As we have seen, there is at most a $1$-round difference between the two versions. We remark that it is possible to show that any algorithm that solves the problem in the original definition must in fact be aware of enough matched neighbors that it can output the required pointers without the $1$-round penalty, implying that these two problems are actually \emph{equivalent}.

We now show how to connect the bipartite $x$-maximal $y$-matching problem (with the encoding described above) to the family of previously defined problems. In particular, we define an intermediate problem, and we apply Theorem \ref{thm:speeduppaper} twice, once in Lemma \ref{lem:varinitone} and once in Lemma \ref{lem:varinittwo}, to show that $\Psi_\Delta^W(x,y,0)$ is exactly two rounds easier than the bipartite $x$-maximal $y$-matching problem. 
The white diagram of the intermediate problem $\Phi'(x,y)$ that we define is represented in Figure \ref{fig:diagramphi2}. For extreme values of $x$ and $y$ additional arrows may be present.

\begin{figure}
	\centering
	\includegraphics[width=0.2\textwidth]{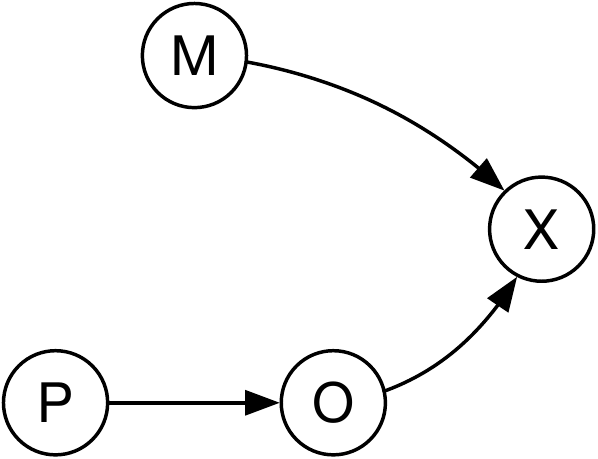}
	\caption{The white diagram of $\Phi'(x,y)$.}
	\label{fig:diagramphi2}
\end{figure}

\begin{lemma}\label{lem:varinitone}
	For any $\Delta \geq 2$, $0 \leq x \leq \Delta$ and $1 \leq y \leq \Delta-1$, the problem $\Phi'(x,y)$ given by the black configurations
	\begin{align*}
		&\X^{y-1} \s \M \s \O^{\Delta-y} \\
		&\X^{y} \s \P^{\Delta-y} \enspace,
	\end{align*}
	and the white configurations
	\begin{align*}
		&[\M \P \O \X]^{y-1} \s [\M \X] \s [\P \O \X]^{\Delta-y} \\
		&[\P \O \X]^{x} \s [\O \X]^{\Delta-x} \enspace.
	\end{align*}
	is a black relaxation of $\REB(\Phi^W_{\Delta}(x,y))$.
\end{lemma}

\begin{proof}
	Similarly to the proof of Lemma \ref{lem:vartech}, we start by showing that any configuration in the black constraint of $\REB(\Phi^W_{\Delta}(x,y))$ can be extended to one of the configurations
	\begin{align*}
		&\bMPOX^{y-1} \s \bM \s \bPOX^{\Delta-y} \\
		&\bMPOX^{y} \s \bOX^{\Delta-y} \enspace.
	\end{align*}
	Consider an arbitrary configuration $\mathcal C$ in the black constraint of $\REB(\Phi^W_{\Delta}(x,y))$.
	We distinguish two cases.
	
	If $\mathcal C$ contains the set $\bM$, then at most $y-1$ of the other $\Delta - 1$ sets in $\mathcal C$ can contain an $\M$ since in each black configuration of $\Phi^W_{\Delta}(x,y)$, there are at most $y$ labels $\M$.
	Hence, the remaining $\Delta-y$ labels must be subsets of $\bPOX$ each, and it follows that $\mathcal C$ can be extended to $\bMPOX^{y-1} \s \bM \s \bPOX^{\Delta-y}$.

	If $\mathcal C$ does not contain the set $\bM$, then each set in $\mathcal C$ cannot contain $\P$ as otherwise we could choose one $\P$ and $\Delta-1$ further labels that are all different from $\M$ from the $\Delta$ sets in $\mathcal C$, which does not yield a black configuration of $\Phi_\Delta^W(x,y)$.
	Using, again, the fact that in each black configuration of $\Phi^W_{\Delta}(x,y)$ there are at most $y$ labels $\M$, it follows that $\mathcal C$ can be extended to $\bMPOX^{y} \s \bOX^{\Delta-y}$.
	
	By applying Lemma \ref{lem:merge} and Observation \ref{obs:strongerset}, we obtain that replacing the black constraint of $\REB(\Phi^W_{\Delta}(x,y))$ by the configurations
	\begin{align*}
		&\bMPOX^{y-1} \s \bM \s \bPOX^{\Delta-y} \\
		&\bMPOX^{y} \s \bOX^{\Delta-y} \enspace,
	\end{align*}
	will result in a black relaxation of $\REB(\Phi^W_{\Delta}(x,y))$.
	Computing the white configurations of the relaxation in an analogous manner to the approach in Lemma \ref{lem:vartech}, we obtain the configurations
	\begin{align*}
		&[\bM \bOX \bPOX \bMPOX]^{y-1} \s [\bM \bMPOX] \s [\bOX \bPOX \bMPOX]^{\Delta-y} \\
		&[\bOX \bPOX \bMPOX]^{x} \s [\bPOX \bMPOX]^{\Delta-x} \enspace.
	\end{align*}
	Finally, renaming the sets in the black and white constraint of the relaxation according to
	\begin{align*}
		\bM &\rightarrow \M \\
		\bOX &\rightarrow \P \\
		\bPOX &\rightarrow \O \\
		\bMPOX &\rightarrow \X \enspace,
	\end{align*}
	shows that the relaxation is identical to $\Phi'(x,y)$, thereby proving the lemma.
\end{proof}

\begin{lemma}\label{lem:varinittwo}
	For any $\Delta \geq 2$, $0 \leq x \leq \Delta$ and $1 \leq y \leq \Delta-1$, problem $\Psi_\Delta^W(y,x,0)$ is a white relaxation of $\REW(\Phi'(x,y))$.
\end{lemma}

\begin{proof}
	Similarly to the proof of Lemma \ref{lem:vartech}, we start by showing that any configuration in the white constraint of $\REW(\Phi'(x,y))$ can be extended to one of the configurations
	\begin{align*}
		&\bMPOX^{y-1} \s \bMX \s \bPOX^{\Delta-y} \\
		&\bMPOX^y \s \bPOX^{x} \s \bOX^{\Delta-y-x} \\
		&\bMPOX^y \s \bX \s \bPOX^{\Delta-y-1} \enspace.
	\end{align*}
	Consider an arbitrary configuration $\mathcal C$ in the white constraint of $\REW(\Phi'(x,y))$.
	We distinguish three cases.
	
	Analogously to the argumentation in the case distinction in the proof of Lemma \ref{lem:vartech}, we obtain that if $\mathcal C$ contains the set $\bX$, then $\mathcal C$ can be extended to $\bMPOX^y \s \bX \s \bPOX^{\Delta-y-1}$, and if $\mathcal C$ contains the set $\bMX$, then $\mathcal C$ can be extended to $\bMPOX^{y-1} \s \bMX \s \bPOX^{\Delta-y}$.
	Hence, consider the last remaining case, i.e., that $\mathcal C$ contains neither the set $\bX$ nor the set $\bMX$.
	Again, analogously to the argumentation in the proof of Lemma \ref{lem:vartech}, we see that each set in $\mathcal C$ must contain the label $\O$, which in turn implies that at most $x$ sets in $\mathcal C$ can contain the label $\P$.
	Since at most $y$ sets in $\mathcal C$ can contain an $\M$, it follows that $\mathcal C$ can be extended to $\bMPOX^y \s \bPOX^{x} \s \bOX^{\Delta-y-x}$.

	Continuing analogously to the proof of Lemma \ref{lem:vartech}, by applying Lemma \ref{lem:merge} and Observation \ref{obs:strongerset}, and computing the black constraint, we obtain a white relaxation of $\REW(\Phi'(x,y))$ that is given by the white configurations
	\begin{align*}
		&\bMPOX^{y-1} \s \bMX \s \bPOX^{\Delta-y} \\
		&\bMPOX^y \s \bPOX^{x} \s \bOX^{\Delta-y-x} \\
		&\bMPOX^y \s \bX \s \bPOX^{\Delta-y-1} \enspace.
	\end{align*}
	and the black configurations
	\begin{align*}
		&[\bX  \bMX \bOX \bPOX \bMPOX]^{y-1} \s [\bMX \bMPOX] \s [\bOX \bPOX \bMPOX]^{\Delta-y} \\
		&[\bX  \bMX \bOX \bPOX \bMPOX]^y \s [\bPOX \bMPOX]^{\Delta-y} \enspace.
	\end{align*}
	Now adding the black configuration $[\bX  \bMX \bOX \bPOX \bMPOX]^y \s [\bMPOX] \s [\bOX \bPOX \bMPOX]^{\Delta-y-1}$ (which can only relax the problem further) and renaming the sets according to
	\begin{align*}
		\bX &\rightarrow \Z \\
		\bOX &\rightarrow \P \\
		\bPOX &\rightarrow \O \\
		\bMX &\rightarrow \M \\
		\bMPOX &\rightarrow \X \enspace,
	\end{align*}
	yields a relaxation that is identical to $\Psi_\Delta^W(y,x,0)$, thereby proving the lemma.
\end{proof}

We now prove that, if the parameters $a$, $b$, and $c$ are not too large, then the problem $\Psi^W_{\Delta}(a,b,c)$  is not $0$-round solvable.

\begin{lemma}\label{lem:0round}
	For any $\Delta \geq 2$, $1 \leq a \leq \Delta-1$, $0 \leq b \leq \Delta-a-1$, and $0 \leq c \leq \Delta-a-1$, there is no deterministic white algorithm that solves $\Psi^W_{\Delta}(a,b,c)$ in $0$ rounds.
\end{lemma}

\begin{proof}
	For a contradiction, assume that such a white algorithm exists.
	As it is deterministic, and in the port numbering model each node has exactly the same information in the beginning, i.e., after $0$ rounds, each white node will necessarily choose the same configuration $\mathcal C$ from the white constraint and output the $\Delta$ labels contained in $\mathcal C$ on the $\Delta$ incident edges according to a fixed bijective function that maps the set of port numbers (or, in other words, the set of incident edges) to the (multi)set of $\Delta$ labels in $\mathcal C$.
	For each edge $e$, call the port number that the white endpoint of $e$ assigns to edge $e$ the \emph{white port of $e$}.
	From the above, it follows that for each label $L$ in $\mathcal C$, there is a fixed port number $\gamma$ such that each edge whose white port equals $\gamma$ will receive the output $L$.
	Therefore, w.l.o.g., we can assume the following.
	If $\mathcal C = \X^{a-1} \s \M \s \O^{\Delta-a}$, then each edge whose white port equals $1$ receives output label $\M$; if $\mathcal C = \X^a \s \O^{b} \s \P^{\Delta-a-b}$, then each edge whose white port equals $1$ receives output label $\P$; if $\mathcal C = \X^a \s \Z \s \O^{\Delta-a-1}$, then each edge whose white port equals $1$ receives output label $\Z$.
	Note that for the second case, our choice of parameters ensures that $\mathcal C$ contains at least one $\P$.

	Now, consider a black node $v$ such that the white port of each edge incident to $v$ equals $1$.
	Clearly, there exist input graphs in which such a node occurs.
	Depending on which of the white configurations was chosen as $\mathcal C$, the multiset consisting of the output labels of the $\Delta$ edges incident to $v$ is $\M^\Delta$, $\P^\Delta$, or $\Z^\Delta$.
	As, due to our choice of parameters, none of these multisets is a black configuration of $\Psi^W_{\Delta}(a,b,c)$, we obtain a contradiction to the correctness of the algorithm, which proves the lemma.
\end{proof}

We finally combine all the results obtained above to prove a lower bound for the bipartite $x$-maximal $y$-matching problem in the port numbering model.
\begin{theorem}\label{thm:pnlower}
	In the port numbering model, for any $0 \leq x \leq \Delta$ and $1 \leq y \leq \Delta - 1$, the white round complexity of bipartite $x$-maximal $y$-matching is at least	
	\begin{align*}
		2 \lceil (\Delta - x)/y\rceil \enspace, \qquad &\text{ if }  \lceil \Delta/y\rceil > \lceil ( \Delta - x)/y\rceil \enspace, \text{ and } \\
		2 \lceil (\Delta - x)/y\rceil - 1 \enspace, \qquad &\text{ if } \lceil \Delta/y\rceil = \lceil ( \Delta - x)/y\rceil \enspace.
	\end{align*}
	In particular, the white round complexity of \bmm{} is at least $2 \Delta - 1$.
\end{theorem}

\begin{proof}
	By Lemmas \ref{lem:vartech}, \ref{lem:varinitone}, and \ref{lem:varinittwo}, we obtain the following sequence of problems in the round elimination framework if we start with bipartite $x$-maximal $y$-matching:
	\[
			\Phi_\Delta^W(x,y) \rightarrow \Phi'(x,y) \rightarrow \Psi_\Delta^W(y,x,0) \rightarrow \Psi_\Delta^B(y,y,x) \rightarrow \Psi_\Delta^W(y,x+y,y) \rightarrow \Psi_\Delta^B(y,2y,x+y) \rightarrow \dots
	\]
	Each problem in the sequence is obtained from the previous problem by applying the function $\REB()$ or the function $\REW()$, and additionally relaxing the obtained problem.
	By Theorem \ref{thm:speeduppaper}, it follows that each problem has a black, resp.\ white complexity that is at least one round less than the white, resp.\ black, complexity of the previous problem, depending on which of the two functions was applied in the respective step.

	With an analogous argumentation to the one in the proof of Lemma \ref{lem:0round}, we can see that $\Phi_\Delta^W(x,y)$ has white round complexity $0$ if and only if $x = \Delta$ or $y = \Delta$, and also $\Phi'(x,y)$ has black round complexity $0$ if and only if $x = \Delta$ or $y = \Delta$.
	As the lemma assumes $y < \Delta$ and only promises a $0$-round lower bound for the case where $x = \Delta$, the lemma statement holds for these cases.
	Hence, assume that $x < \Delta$ (and $y < \Delta$).

	Consider first the case that $\lceil \Delta/y\rceil > \lceil ( \Delta - x)/y\rceil$, i.e., $\lceil ( \Delta - x)/y\rceil \leq \lceil \Delta/y\rceil - 1$.
	Due to the growth behavior of the parameters in the above sequence promised by Lemma \ref{lem:vartech}, the problem $\Psi_\Delta^B(y,b,c)$ obtained after $2 \lceil (\Delta - x)/y\rceil - 1$ steps, starting from $\Phi_\Delta^W(x,y)$, satisfies 
	\begin{align*}
		b &= \left(\lceil \left(\Delta - x\right)/y\rceil - 1\right) \cdot y \leq \left(\left(\lceil \Delta/y\rceil - 1\right) - 1\right) \cdot y \leq \Delta - 1 - y \enspace, \qquad \qquad \text{and}\\
		c &= x + \left(\lceil \left(\Delta - x\right)/y\rceil - 2\right) \cdot y \leq x + \left( \Delta - x - 1 - y \right) = \Delta - 1 - y \enspace.
	\end{align*}
	By Lemma \ref{lem:0round}, we obtain that, for these parameters $b$ and $c$, problem $\Psi_\Delta^B(y,b,c)$ cannot be solved in $0$ rounds by a black algorithm, which implies that the complexity of bipartite $x$-maximal $y$-matching must be at least $2 \lceil (\Delta - x)/y\rceil$.

	Now, consider the case that $\lceil \Delta/y\rceil = \lceil ( \Delta - x)/y\rceil$.
	Similarly to the previous case, we obtain that the problem $\Psi_\Delta^W(y,b,c)$ in the above sequence obtained after $2 \lceil (\Delta - x)/y\rceil - 2$ steps satisfies
	\begin{align*}
		b &= x + \left(\lceil \left(\Delta - x\right)/y\rceil - 2\right) \cdot y \leq x + \left( \Delta - x - 1 - y \right) = \Delta - 1 - y \enspace, \qquad \qquad \text{and}\\
		c &= \left(\lceil \left(\Delta - x\right)/y\rceil - 2\right) \cdot y \leq \left(\lceil \Delta/y\rceil - 2 \right) \cdot y \leq \Delta - 1 - y \enspace.
	\end{align*}
	Again, by Lemma \ref{lem:0round}, we obtain that, for these parameters $b$ and $c$, problem $\Psi_\Delta^W(y,b,c)$ cannot be solved in $0$ rounds by a white algorithm, which implies that the complexity of bipartite $x$-maximal $y$-matching must be at least $2 \lceil (\Delta - x)/y\rceil - 1$.
\end{proof}

Note that for the case $y = \Delta$, there is a trivial $0$-round upper bound, which is why we do not cover this case in Theorem \ref{thm:pnlower}.

\section{Upper bounds}\label{sec:upper}
We will now prove upper bounds for the family of bipartite $x$-maximal $y$-matching problems. The algorithm we provide is very similar to the \emph{proposal} algorithm of \citet{Hanckowiak1998}, that works as follows: for $i=1,\ldots,\Delta$, at round $2i-2$, white nodes propose to be matched to their $i$-th neighbor (if it exists), and black unmatched nodes accept a proposal and reject all the others. This algorithm requires $2\Delta$ rounds, but it can be made $1$ round faster by making black nodes propose and white nodes output.
For the special case of \bmm{}, i.e., bipartite $0$-maximal $1$-matching, the algorithm that we propose achieves the same running time of the proposal algorithm.
The main difference is that, in the general case, our algorithm achieves a better \emph{intermediate} state than the proposal algorithm, i.e., the partial solution that it maintains at each step is better than the one maintained by the original proposal algorithm, and this allows us to get \emph{exact} bounds for the whole family of problems. In the following, we will consider only the case where $1 \le y < \Delta$, since for $y = \Delta$ there exists a trivial $0$-rounds algorithm that simply puts each edge into the matching.

The high level idea of the algorithm is the following. White nodes start sending matching proposals to $y$ arbitrary black neighbors. Black nodes that receive proposals accept one and reject the others. Black nodes that did \emph{not} receive any proposal, send proposals to $y$ white neighbors that have not previously sent a proposal to the black node (but can be chosen arbitrarily apart from that). White nodes that receive proposals and are not matched, accept one and reject the others. By repeating this procedure for a sufficiently large number of rounds, we can ensure that if a black or white node is still unmatched, then a large number of neighbors are already matched. This is the main difference with the standard proposal algorithm, where only one side proposes: if we stop the execution of the algorithm before its natural termination and we look at the partial solution that it maintains, we see that the partial solution of the standard proposal algorithm maintains good guarantees for one side only.

We will now formally describe the procedure. Each node keeps track of the state of each incident edge, by maintaining four sets:
\begin{itemize}
	\item The set $F$ contains all free edges, i.e., the edges over which no request has ever been sent or received.
	\item The set $M$ contains all edges that are already part of the matching.
	\item The set $S$ contains all edges where the node \emph{sent} a proposal.
	\item The set $R$ contains all edges where the node \emph{received} a proposal.
\end{itemize}
Note that the design of the algorithm will ensure that each edge will be used for at most one proposal in total, which implies that the set $S$ and the set $R$ are disjoint.
In the beginning, each node initializes the set $F$ with all incident edges, and all the other sets as empty sets. Then, nodes apply the following $1$-round procedure repeatedly. Each round the role of active and passive is reversed between black and white nodes. Initially, \emph{white} nodes are active.
\begin{enumerate}
	\item Active nodes with $M = \emptyset$ do the following:
	\begin{enumerate}
		\item If $R$ contains at least one edge $e$, send an acceptance over the edge $e$, and put $e$ into the set $M$.
		\item Otherwise, remove $y$ edges from the set $F$ (or \emph{all} edges if $|F| < y$), add them to the set $S$, and send proposals over these edges.
	\end{enumerate}
	\item Passive nodes with $M = \emptyset$ do the following:
	\begin{enumerate}
		\item If acceptances are received over a set of edges $E$, add elements of $E$ to $M$ 
		\item If requests are received over a set of edges $E$, add elements of $E$ to $R$.
	\end{enumerate}
\end{enumerate}
If at some point $M \neq \emptyset$ for some node $v$, then $v$ terminates.

In the following, we prove that, by applying this procedure $2k$ times, we obtain a solution (known by white nodes) for the bipartite $x$-maximal $y$-matching problem, where $x = \max \{ 0,\Delta - k y \}$. 

Regarding the maximum number of incident edges in the matching, note that a node sends at most $y$ proposals in parallel, and if at least one is accepted, then the node stops participating in the following phases. Thus, a node is matched with at most $y$ neighbors.

Next, we show that, for each node $v$, if $M$ is empty after $2k$ steps, then $\min\{ \deg(v), ky \}$ neighbors of $v$ are matched (i.e., have non-empty $M$). Notice that a node that is still unmatched after $2k$ steps has been active for $k$ steps, and in those steps it added a total of $\min\{ \deg(v), ky \}$ different edges to the set $S$, and hence sent $\min \{ \deg(v), ky \}$ proposals.  If all the proposals are rejected, it means that the receivers of those proposals are matched with different nodes, and thus if a node $v$ is unmatched, then at least $\min\{ \deg(v), ky \}$ neighbors are matched. Notice that black nodes may not know if they are matched or not: the proposals that they sent in round $2k$ are received by white nodes that output the solution without informing black nodes about their decisions.

Hence, if we want to solve the bipartite $x$-maximal $y$-matching problem, we can set $k=\lceil \frac{\Delta-x}{y}\rceil$ and obtain a running time of $2 \lceil \frac{\Delta-x}{y}\rceil$. In order to satisfy the formal specifications of (the encoded version of) the bipartite $x$-maximal $y$-matching problem, white matched nodes output $\M$ on matched edges and $\O$ on all the others, while white unmatched nodes output $\P$ on all the edges where they sent a proposal and $\X$ on all the others. This clearly bounds the number of $\X$ given as output by white nodes correctly. Concerning black nodes, since an $\X$ cannot be given as output by a white node on an edge where a black node sent a proposal, the number of $\X$ incident to black nodes is also correctly bounded.

Finally, note that if $ \lceil \Delta/y\rceil = \lceil ( \Delta - x)/y\rceil$ then we can save one round. In fact, in this case we can solve the harder case where $x=0$ and thus obtain strict maximality. For this purpose, we can use the previously described algorithm (setting $x=0$) and modify it as follows: \emph{black} nodes start and \emph{white} nodes output. Notice in fact that if we run the original algorithm for just $2k-1$ rounds, instead of obtaining the maximality guarantee for both sides, we obtain the maximality guarantee for one side only. But for $x=0$, this is enough, since if one side is matched or proposed to \emph{all} neighbors, then there is no reason for the other side to send additional proposals.

Thus, the bipartite $x$-maximal $y$-matching problem can be solved by a white algorithm in
\begin{align*}
2 \lceil (\Delta - x)/y\rceil \enspace, \qquad &\text{ if }  \lceil \Delta/y\rceil > \lceil ( \Delta - x)/y\rceil \enspace, \text{ and } \\
2 \lceil (\Delta - x)/y\rceil - 1 \enspace, \qquad &\text{ if } \lceil \Delta/y\rceil = \lceil ( \Delta - x)/y\rceil \enspace.
\end{align*}

\section{Lifting the lower bounds to the LOCAL model}\label{sec:lifting}
In the previous sections we proved lower and upper bounds for the port numbering model. While upper bounds for the port numbering model are trivially also upper bounds for the LOCAL model, the same is not true for lower bounds. Hence, in this section we show how to lift the obtained lower bounds to the LOCAL model. The proof ideas that we use are heavily inspired by the proofs in \cite{Balliu2019}.

In Sections \ref{sec:lower} and \ref{sec:upper} we proved that, on graphs with maximum degree $\Delta$, the bipartite $x$-maximal $y$-matching problem in the port numbering model requires exactly
\[
T_\Delta(x,y) =
\begin{cases}
2 \lceil (\Delta - x)/y\rceil, ~ &\text{ if }  \lceil \Delta/y\rceil > \lceil ( \Delta - x)/y\rceil\\
2 \lceil (\Delta - x)/y\rceil - 1, ~ &\text{ if } \lceil \Delta/y\rceil = \lceil ( \Delta - x)/y\rceil
\end{cases}
\]
rounds. In order to show this result, we applied the round elimination technique for exactly $T_\Delta(x,y) - 1$ times, and at each time we performed simplifications that reduced the number of labels to $5$. We thus obtained a sequence of problems $\Pi^{\Delta,x,y}_0, \ldots, \Pi^{\Delta,x,y}_{T_\Delta(x,y)-1}$, where the first problem $\Pi^{\Delta,x,y}_0$ is the initial problem $\Phi_\Delta^W(x,y)$, the second is $\Phi'(x,y)$, the third is $\Psi_\Delta^W(y,x,0)$, and the last is a problem that is still not solvable in $0$ rounds (in the following we will omit the superscript and just write $\Pi_0, \ldots, \Pi_{T-1}$). We will not consider the exact problem sequence of Section \ref{sec:lower}, where for problems with even index a white algorithm is considered and for problems with odd index a black algorithm is considered. Instead, for simplicity, we will only consider white algorithms, and we will consider a new sequence of problems that can be obtained by starting from the original sequence and swap black and white constraints on problems with odd indexes. Notice that this is equivalent, by just reversing the roles of black and white nodes on problems with odd indexes. Also, let $S(\Pi)$ be the problem obtained by swapping the black and the white constraints of $\Pi$.

In order to get a lower bound for the LOCAL model, we first show how to obtain a \emph{randomized} lower bound for the port numbering model, and then we argue about what this bound implies for the LOCAL model.

In order to obtain a randomized lower bound for the port numbering model, the only thing that we need to do is to replace the application of the round elimination theorem with a randomized version of it: additionally to the problem sequence (that is the same as in the deterministic case), we also have \emph{local failure probabilities}\footnote{While a \emph{global} failure probability of at most $1/n$ implies that with probability at least $1-1/n$ \emph{no} node of the graph fails, a \emph{local} failure probability of $1/n$ just requires that \emph{each} node fails with probability at most $1/n$.} assigned to the problems in the sequence. In particular, we will consider the sequence of problems $\Pi_0, \ldots, \Pi_{T-1}$ and prove that if we assume that $\Pi_0$ can be solved in $T-1$ rounds with some small local failure probability, then we obtain that $\Pi_{T-1}$ can be solved in $0$ rounds with some larger but still small enough local failure probability. We will then prove that this last problem cannot be solved with small local failure probability, hence obtaining a contradiction with the initial assumption.

We do not need to prove a randomized version of the round elimination theorem: in \cite{binary} it has been shown that for each subsequent problem in the sequence, we can give a bound on the local failure probability as a function of $\Delta$, the number of labels used to describe the previous problem, and the bound on the local failure probability for the previous problem, i.e., we do not have to take into account the black and white constraints, and we can hence use this result in a blackbox manner. In particular, \cite{binary} showed the following result (rephrased for our purposes):
\begin{lemma}[Lemma 41 of \cite{binary}]
	Let $A$ be a white (resp.\ black) randomized $t$-round algorithm for $\Pi$ with local failure probability at most $p$ in the port numbering model. Then there exists a black (resp. white) randomized $(t-1)$-round algorithm $A'$ for  $~\REB(\Pi)$ (resp. $\REW(\Pi)$), with local failure probability at most $2^{1/(\Delta+1)}(\Delta |\Sigma_\Pi|)^{\Delta / (\Delta + 1)} p^{1/(\Delta+1)} + p$.
\end{lemma}
We can use this lemma to incorporate the given local failure probability analysis into the sequence of problems $\Pi_i$ that we obtained. Since $|\Sigma_{\Pi_i}| \le 5$ for all problems in the family, we can rewrite the upper bound on the local failure probability as
\[
2^{1/(\Delta+1)}(\Delta |\Sigma_\Pi|)^{\Delta / (\Delta + 1)} p^{1/(\Delta+1)} + p \le
K \Delta p^{1/(\Delta+1)} 
\]
for some constant $K$. Also, since we can reverse the roles of black and white nodes, we obtain the following corollary:
\begin{corollary}\label{cor:singlestep}
		Let $A$ be a white randomized $t$-round algorithm for $\Pi$ with local failure probability at most $p$ in the port numbering model. Then there exists a white randomized $(t-1)$-round algorithm $A'$ for  $~S(\REB(\Pi))$ with local failure probability at most $K \Delta p^{1/(\Delta+1)}$, for some constant $K$.
\end{corollary}

We will now prove the following lemma, that bounds the local failure probability after multiple round elimination steps.
\begin{lemma}\label{lem:multiplesteps}
			Let $A$ be a white $t$-round randomized algorithm for $\Pi_0$ with local failure probability at most $p$ in the port numbering model. Then there exists a white randomized $(t-j)$-round algorithm $A'$ for  $\Pi_j$ with local failure probability at most $(K\Delta)^2 p^{1/ (\Delta+1)^j}$ for some constant $K$, for all $j \le t$.
\end{lemma}
\begin{proof}
We can apply Corollary \ref{cor:singlestep} multiple times to get an upper bound on the local failure probability after $j$ steps of round elimination. We can prove by induction that the failure probability $p_j$ after $j$ steps is at most 
\[
	(K\Delta)^2 p^{1/ (\Delta+1)^j} \enspace.
\]
The base case holds trivially. For the induction step we have that
\[
 p_{j+1} \le K \Delta p_j^{1/(\Delta+1)} \enspace.
\]
By applying the induction hypothesis, we get that $p_{j+1}$ is at most
\[
 K \Delta ( (K\Delta)^2 p^{1/ ( \\\Delta+1)^j} )^{1/(\Delta+1)} \le (K \Delta)^{1+2/(\Delta+1)} p^{1/ (\Delta+1)^{j+1}} \le  (K \Delta)^{2} p^{1/ (\Delta+1)^{j+1}} \enspace.
\]
\end{proof}
We now lower bound the local failure probability of a $0$ rounds algorithm solving $\Pi_{T\\-1}$.
\begin{lemma}\label{lem:basecase}
	There is no white randomized algorithm that solves $\Pi_{T-1}$ in $0$ rounds with local failure probability $p \le 1 / \Delta^{2\Delta}$ in the port numbering model.
\end{lemma}
\begin{proof}	
	To show this lemma we will prove a stronger statement, namely, that for \emph{any} problem that does not admit a white \emph{deterministic} algorithm in the port numbering model, there is no white randomized algorithm solving it with local failure probability smaller than $1/\Delta^{2\Delta}$ in the port numbering model. Let $\Pi = (W,B)$ be a problem having white constraint $W$ and black constraint $B$ using labels over the set $L$. Let $W = \{c_1,\ldots,c_k\}$. Notice that for the family of problems under consideration $k=O(\Delta)$, and $|L|=5$.
	
	Any $0$ rounds algorithm is just a probability assignment to each $c_i$, that is, any $0$ rounds algorithm must output $c_i$ with some probability $0 \le p_i \le 1$, such that $\sum p_i = 1$. By the pigeonhole principle there must exist a configuration $\bar{c} = c_i$ such that $p_i \ge 1/k$. Hence, all white nodes output the configuration $\bar{c}$ with probability at least $1/k$. Also, conditioned on the fact that a white node outputs the configuration $\bar{c}$, for each label $\ell \in L$ that appears in $\bar{c}$ there must be a port number $j$ where $\ell$ is written with probability at least $1/\Delta$. 
	
	Since $\Pi$ is not $0$ rounds solvable in the port numbering model, then there exists a configuration $\bar{b} = \{\ell_1,\ldots,\ell_\Delta\}$ that is \emph{not} in $B$, such that $\ell_j \in \bar{c}$ for all $j$. Let us now consider a black node and let us connect $\Delta$ white nodes $w_1,\ldots,w_\Delta$ to it in some specific way. Node $w_i$ will output $\ell_i$ on some port $j$ with probability at least $\frac{1}{k\Delta}$. We connect port $j$ to port $i$ of the black node. Thus, the black node will be incident to the bad configuration $\bar{b}$ with probability at least
	\[
		\frac{1}{(k\Delta)^\Delta}
	\]
	 Since in our case the white constraint contains $O(\Delta)$ configurations (actually $O(\Delta)$ for the first problem of the sequence and then $3$ for all the others), for large enough $\Delta$ the claim follows.
\end{proof}
We now lower bound the local failure probability of a randomized algorithm for $\Pi_0$ running in $t = T - 1$ rounds, in order to show that if an algorithm for $\Pi_0$ runs ``too fast'' then it must fail with large probability.

By applying Lemma \ref{lem:multiplesteps} and setting $j = t = T-1$ we get that if there is a $(T-1)$-algorithm for $\Pi_0$ fails with probability at most $p$, then there is a $0$-round algorithm for $\Pi_{T-1}$ that fails with probability at most $(K\Delta)^2 p^{1/ (\Delta+1)^{t}}$ for some constant $K$. Now, by applying Lemma \ref{lem:basecase} we obtain that
\[
(K\Delta)^2 p^{1/ (\Delta+1)^{t}} \ge 1 / \Delta^{2\Delta}
\]
For large enough $\Delta$, we get that
\[
p \ge \frac{1}{((K\Delta)^2 \Delta^{2\Delta})^{(\Delta+1)^{t}}} \ge \frac{1}{2^{\Delta^{2t+3}}} 
\]
Hence, we obtain the following result:
\begin{lemma}\label{lem:failure}
	A white randomized algorithm for bipartite $x$-maximal $y$-matching for graphs of maximum degree $\Delta$ that runs in strictly less than $T=T_\Delta(x,y)$ rounds must fail with probability at least $\frac{1}{2^{\Delta^{2T+1}}}$ in the port numbering model.
\end{lemma}

We are now ready to prove lower bounds for the LOCAL model.
Our goal will be to show lower bounds for randomized algorithms in the LOCAL model, for all $\Delta \le \bar{\Delta}$, where we will try to make $\bar{\Delta}$ as large as possible, as a function of $n$.
Since in the port numbering model we can easily generate unique IDs with high probability of success, a lower bound for this model directly implies a lower bound for the LOCAL model. Thus, the only thing that we need to take care of is to show the largest possible lower bound, as a function of $\Delta$ and $n$, that does not violate the previous assumptions, that is, in order to apply the round elimination technique our lower bound graph family must contain a large enough tree-like neighborhood, and we need to choose the right value of $\Delta$ such that we can apply Lemma \ref{lem:failure} (in particular, if $\Delta$, as a function of $n$, is too large, Lemma \ref{lem:failure} would not imply large failure probability).

By definition, any randomized algorithm in the LOCAL model has a global failure probability of at most $1/n$, hence also each node must have a local failure probability of at most $1/n$. In order to get the best possible lower bound in terms of $n$, we can choose the largest possible value of $\Delta$ (as a function of $n$) such that:
\begin{itemize}
	\item There is a graph containing a tree-like $\Delta$-regular neighborhood of radius $T_\Delta(x,y) + 1$ (a necessary condition to apply the round elimination theorem).
	\item By applying Lemma \ref{lem:failure} we get a local failure probability that is larger than $1/n$ (implying that $T_\Delta(x,y) -1$ rounds are not enough to solve the problem with local failure probability at most $1/n$).
\end{itemize}
We will prove the following lemma, that states that bipartite $x$-maximal $y$-matchings that are not too relaxed are hard in the LOCAL model.
\begin{lemma}\label{lem:randomized}
	For any $k \ge 1$, and for large enough $\Delta$ and $n$, a white randomized algorithm for bipartite $x$-maximal $y$-matching that fails with probability at most $1/n$ in the LOCAL model requires at least $T_\Delta(x,y)$ rounds, unless $T_\Delta(x,y) \ge \frac{1}{3k} \log\log n / \log \log \log n$, or $T_\Delta(x,y) \le \Delta^{1/k}$.
\end{lemma}
\begin{proof}
	Consider any $\Delta$ satisfying $T=T_\Delta(x,y) <  \frac{1}{3k}\frac{\log \log n}{\log \log \log n}$, that is, $\Delta > \big(\frac{1}{3k}\frac{\log \log n}{\log \log \log n}\big)^k$.
	
	Assume for a contradiction that a faster algorithm exists, that is, it runs in $t < T_\Delta(x,y)$ rounds such that $t < \frac{1}{3k} \frac{\log \log n}{\log \log \log n}$. We can construct a graph where there is a neighborhood that locally looks like a tree up to distance $T+1$ (since for large enough $n$, $\Delta^{T+1} < n$), and by Lemma \ref{lem:failure} nodes there must fail with probability at least
	\[
		\frac{1}{2^{\Delta^{2T+1}}} \ge \frac{1}{2^{2^{3T\log \Delta}}}  \ge \frac{1}{2^{2^{3k\frac{1}{3k}\frac{\log \log n}{ \log \log \log n} \log (\frac{1}{3k}\frac{\log \log n}{\log \log \log n})}}}  \ge \frac{1}{2^{2^{\log \log n}}}  \ge \frac{1}{n}
	\]
	 for large enough $n$.
\end{proof}

This result directly implies the same lower bound for deterministic algorithms in the LOCAL model. Nevertheless, we can prove deterministic lower bounds for even larger values of $n$. For this purpose, we use the same idea as in \cite[Theorem 25]{Balliu2019}, that is, we assume to have a too fast deterministic algorithm, and we use it to define a too fast randomized algorithm, that would contradict the previously shown lower bound. Unfortunately, while doing so, we lose the tight-in-$\Delta$ dependency, but notice that this happens only for $n = \Omega(\frac{\log \log n}{\log\log\log n})$, since for smaller values the previous result applies.
\begin{lemma}\label{lem:det}
		For any $k \ge 1$, and for large enough $\Delta$ and $n$, a white deterministic algorithm for bipartite $x$-maximal $y$-matching in the LOCAL model requires $\Omega(\min(T_\Delta(x,y),\frac{1}{k} \log n / \log \log n)$ rounds, unless $T_\Delta(x,y) \le \Delta^{1/k}$.
\end{lemma}
\begin{proof}
	Assume that there is a faster algorithm $A$ that runs in $t = t(n,\Delta) < \min(T_\Delta(x,y),\frac{1}{4k} \log n / \log \log n))$ rounds; we will prove that this leads to a contradiction. Let us fix $\Delta$ as in the proof of Lemma \ref{lem:randomized}, that is, $\Delta  > \big(\frac{1}{3k} \frac{\log \log n}{\log \log \log n}\big)^k$. We now set $N = \log n$. The goal is to execute $A$ by lying about the size of the graph, claiming it is of size $N$. Notice that this gives a running time $\bar{t} = t(N,\Delta) \le \frac{1}{4k} \log N / \log \log N \le \frac{1}{4k} \log \log n / \log \log \log n$. In order to do so, we need to compute a new small enough ID assignment, such that $A$ cannot detect the lie, that is, the ID assignment must be locally valid: in each radius $\bar{t}$ neighborhood the new ID assignment must be unique. For this purpose, we compute a $c=O(\Delta^{2\bar{t}+2}(\log \log N + \log \Delta^{2\bar{t}+2}))$-coloring of $G^{2\bar{t}+2}$, the $(2\bar{t}+2)$-th power of $G$, in $O(\bar{t})$ rounds \cite[Theorem 25]{Balliu2019}. The algorithm $A$ cannot detect the lie, since $c = O(\log^{\frac{3}{4}}n) \cdot \log^{o(1)} n < N$ for large enough $n$.
	
	Thus we can now run $A$ and get a solution for the problem in $o(T_\Delta(x,y) + \log N / \log \log N)$ rounds, and since $N = \log n$, we get a contradiction with Lemma \ref{lem:randomized}, since we used the same graph family.
\end{proof}

Both Lemma \ref{lem:randomized} and Lemma \ref{lem:det} assume $\Delta$ to be larger than some constant. Notice that tight-in-$\Delta$ bounds can be actually obtained in the LOCAL model also for smaller values of $\Delta$, as we needed large values of $\Delta$ only to be able to obtain good in $n$ dependencies. In fact, for fixed values of $\Delta$, we need to apply the randomized round elimination theorem only a constant number of times, and thus we obtain that an algorithm for $\Pi_{T-1}$ that runs in $0$ rounds must fail with local failure probability $f(n) = o(1)$; hence by taking large enough values of $n$ we still get a contradiction, since any $0$ rounds algorithm for $\Pi_{T-1}$ must fail with constant probability, if $\Delta$ is constant. In particular, for any combination of integers $x \geq 0$, $y \geq 1$, $\Delta \geq 2$, our bound of $T_\Delta(x,y)$ is truly tight.

\section{Behind the scenes}\label{sec:backstage}
In the introduction we claimed that in some sense, \emph{bounded} automatic round elimination can be automated in practice.
We will discuss now how this automatization works, using the example of (bipartite) \mm{}.

Concerning lower bounds, we can actually feed \bmm{}, for \emph{small} values of $\Delta$ (such as $3$, $4$ and $5$), to a computer program and get a lower bound of $2\Delta-1$ automatically \cite{Olivetti2019}.
This does not directly give a proof for all values of $\Delta$, but it gives good \emph{ideas} that can help us to prove a lower bound for \emph{all} values of $\Delta$.
The important output we obtain from a computer program in this respect is a sequence of simplifications for each problem $\Pi'$ obtained after a round elimination step.
For each $\Pi'$, ideally these simplifications transform $\Pi'$ into a problem $\Pi^*$ that is just a tiny bit easier to solve, but much simpler to describe.

Sometimes, the obtained simplifications are very specific to the chosen value of $\Delta$, and thus they do not generalize, i.e., they do not help in extracting general relaxation patterns for arbitrary $\Delta$.
For instance, this is what happens by feeding \emph{non-bipartite} \mm{} to the program: we get something very different for the limited small values of $\Delta$ that the program can handle, and it is quite hard to see the pattern.
However, at other times, there is some common pattern in the simplification that the program performs for different values of $\Delta$, and by understanding those we can then try to design a general proof.
This is exactly what happens with \emph{bipartite} \mm{}. 

This, in turn, allowed us to handle the entire family of problems of bipartite $x$-maximal $y$-matchings: we could \emph{not} feed this whole family of problems to a computer program and get an automatic lower bound, but after having obtained a \emph{deep} understanding of \bmm{}, we managed to \emph{manually} design the right simplifications that worked for this much larger family of problems.

Concerning upper bounds, by feeding specific variants of the bipartite $x$-maximal $y$-matching problem to a computer program (for small $\Delta$, $x$ anx $y$), we noticed that we could obtain an automatic upper bound via round elimination by using only \emph{three} labels at each step. These upper bounds were matching the \emph{exact} round complexity that we obtained as lower bounds; before that we thought that the lower bounds are not tight since by modifying the standard \emph{proposal algorithm} to handle bipartite $x$-maximal $y$ matchings we could only obtain upper bounds that were a factor of \emph{two} larger than the lower bounds. Unfortunately, while automatic lower bounds can be somehow understood by a human being, automatic upper bounds obtained via round elimination are still eluding human interpretation. Thus, at this point we knew that some better upper bound exists, but we could not understand it, much less generalize it. Still, the mere fact of \emph{knowing} that there is a better upper bound really helped us to finally find it.

\section{Open Problems}
The round elimination technique allows us to define a problem that is exactly one round easier than the given one. Unfortunately, this comes with a high cost: the description complexity of the problem that we obtain can be exponentially larger than the one from which we start. Most of the lower bounds obtained using this technique can be grouped in two types:
\begin{itemize}
	\item Those where the number of labels grows really fast, like in the case of $3$-coloring \cite{Linial1992}, and weak colorings \cite{Brandt2019,Balliu2019hardness}.
	\item Those where the number of labels is kept constant at each step, like in the case of sinkless orientation \cite{Brandt2016} and \bmm{} \cite{Balliu2019}.
\end{itemize}
Bounded automatic round elimination makes it easy to prove lower bounds of the second type, and understanding for which kind of problems it is possible to obtain lower bounds in this way is a major open question. More concretely, we think that addressing the following open question would help to make significant progress in the process of proving distributed lower bounds. 
\begin{op}
	Characterize the dependency between the maximum number of labels and the best lower bound that can be achieved by using bounded automatic round elimination.
\end{op}

While this question is admittedly (and intentionally) vague, we give a more concrete open question in the following. It seems to be significantly harder to prove lower bounds by using the round elimination technique on non-$2$-colored graphs. We informally mention that a lower bound for non-bipartite \mm{} can be actually obtained by using bounded automatic round elimination, but the lower bound looks significantly different from the one presented in this paper. For \mm{} we can use a bipartition, but for some other problems we cannot, such as the maximal independent set (MIS) problem, as the problem is trivial in such a setting. While a lower bound for MIS is directly implied by the maximal matching lower bound of \citet{Balliu2019}, a lower bound obtained directly via round elimination is not known. Understanding whether even a bounded number of labels might be enough to obtain the same lower bound would be an important step forward.
\begin{op}
	Prove or disprove that the best lower bound for MIS achievable by automatic round elimination is asymptotically larger than the best lower bound achievable by using bounded automatic round elimination.
\end{op}

While for most of the problems it is it true that the problem obtained by applying one step of round elimination is much larger, some other problems, called \emph{fixed points}, are an exception: they constitute a fascinating special case of bounded automatic round elimination in which the number of labels does not grow at all, allowing us to apply this technique as many times as we want. Most of the $\Omega(\log n)$ lower bounds for the LOCAL model have been shown by relaxing the given problem and obtaining a fixed point, with one clear exception \cite{BalliuHLOS19}. We think that understanding when and how it is possible to obtain fixed points would help a lot in better understanding the $\Omega(\log n)$ complexity spectrum.
\begin{op}
	Prove or disprove that for all problems that require $\Omega(\log n)$ rounds in the LOCAL model it is possible to prove such a lower bound by obtaining a fixed point via round elimination and/or suitable relaxations.
\end{op}

\urlstyle{same}
\bibliographystyle{ACM-Reference-Format}
\bibliography{simpler-mm-lb}
\end{document}